\documentclass[a4paper,11pt]{article}
\usepackage[utf8]{inputenc}

\usepackage{geometry}

\usepackage{amsmath,amsthm,amssymb}

\newtheorem{theorem}{Theorem}
\newtheorem{definition}{Definition}

\newtheorem{lemma}{Lemma}
\newtheorem{corollary}{Corollary}

\newtheorem{example}{Example}[section]

\usepackage[hidelinks]{hyperref}
\usepackage{xspace}
\usepackage{multicol}
\usepackage[capitalize]{cleveref}
\crefname{claim}{Claim}{Claims}
\Crefname{claim}{Claim}{Claims}
\crefname{observation}{Observation}{Observations}
\Crefname{observation}{Observation}{Observations}

\usepackage{caption}
\usepackage{multirow}
\usepackage{bm}
\usepackage{dsfont}
\usepackage{fullpage}
\usepackage{thm-restate}

\newcommand{\threefield}[3]{$#1\mid#2\mid#3$}
\newcommand{\prob}{\threefield{1}{}{\sum w_jU_j}\xspace}
\newcommand{\probVariant}{\threefield{1}{}{\sum p_jU_j}\xspace}

\newcommand{\minmakespanhighmultiplicity}{\threefield{P}{\text{HM}}{C_{\max}}\xspace}

\newcommand{\mcc}{\textsc{$k$-Multicolored Clique}\xspace}

\usepackage{tikz}
\usetikzlibrary{fit,calc}
\usetikzlibrary{decorations.pathmorphing}
\tikzset{snake it/.style={decorate, decoration=snake}}

\usepackage{enumitem}
\usepackage{xargs}
\newcommandx{\decprob}[6][3=Input,5=Question]{\vspace{0.125em}\begin{samepage}\begingroup
\label{DEACTIVATEDprob:#2}{
{\noindent \textsc{#1}}}
  \nopagebreak[4]\nopagebreak[4]\vspace{-0.125em}
  \par\noindent\hangindent=\parindent\textbf{#3}:  #4\nopagebreak[4]
  \par\noindent\hangindent=\parindent\textbf{#5}:  #6
  \par\medskip\endgroup\end{samepage}
}

\newcounter{dummycounter}
\makeatletter
\newcommand{\mylabel}[2]{\refstepcounter{dummycounter}#2\def\@currentlabel{#2}\label{#1}}
\makeatother

\author{Klaus Heeger\footnote{Department of Industrial Engineering and Management, Ben-Gurion~University~of~the~Negev, Beer-Sheva, Israel. \texttt{heeger@post.bgu.ac.il}.} \and
Danny Hermelin\footnote{Department of Industrial Engineering and Management, Ben-Gurion~University~of~the~Negev, Beer-Sheva, Israel. \texttt{hermelin@bgu.ac.il}.}}

\title{Minimizing the Weighted Number of Tardy Jobs is W[1]-hard\footnote{Supported by the ISF, grant No.~1070/20.}}

\date{}

\begin{document}
\maketitle

\begin{abstract}
We consider the \prob problem, the problem of minimizing the weighted number of tardy jobs on a single machine. This problem is one of the most basic and fundamental problems in scheduling theory, with several different applications both in theory and practice. We prove that \prob is W[1]-hard with respect to the number~$p_{\#}$ of different processing times in the input, as well as with respect to the number $w_{\#}$ of different weights in the input. This, along with previous work, provides a complete picture for \prob from the perspective of parameterized complexity, as well as almost tight complexity bounds for the problem under the Exponential Time Hypothesis (ETH). 
\end{abstract}
\smallskip


\section{Introduction}
\label{sec:intro}%

In this paper we consider the following fundamental scheduling problem: We are given a set of $n$ jobs $\{x_1,\ldots,x_n\}$, where each job~$x$ is defined by three integer-valued \emph{characteristics:} A \emph{processing time} $p(x) \in \mathbb{N}$, a \emph{weight} $w(x) \in \mathbb{N}$, and a \emph{due date} $d(x) \in \mathbb{N}$. We have a single machine to process all jobs $\{x_1,\ldots,x_n\}$ non-preemptively. Thus, in this setting a \emph{schedule} for $\{x_1,\ldots,x_n\}$ is a permutation $\Pi: \{x_1,\ldots,x_n\} \to \{1,\ldots,n\}$ that specifies the processing order of each job. In this way, we schedule in~$\Pi$ a job~$x$ \emph{starting at time} $R(x)=\sum_{\Pi(y) < \Pi(x)} p(y)$; that is, the total processing time of jobs preceding $x$ in $\Pi$. The \emph{completion time} $C(x)$ of~$x$ is then defined by $C(x)=R(x)+p(x)$.
Job $x$ is said to be \emph{tardy} in $\Pi$ if $C(x) > d(x)$, and \emph{early} otherwise. Our goal is to find a schedule~$\Pi$ where the total weight of tardy jobs is minimized. Following Graham~\cite{Graham1969}, we denote this problem by \prob.

The \prob problem models a very basic and natural scheduling scenario, and is thus very important in practice. However, it also plays a prominent theoretical role, most notably in the theory of scheduling algorithms. For instance, it is one of the first scheduling problems shown to be NP-hard, already included in Karp's famous initial list of 21 NP-hard problems~\cite{Karp72}. The algorithm by Lawler and Moore~\cite{LawlerMoore} which solves the problem in $O(Pn)$ or $O(Wn)$ time, where~$P$ and $W$ are the total processing times and weights of all jobs, is one of the first examples of pseudo-polynomial dynamic programming (see~\cite{HermelinMS2024} for recent improvements on this algorithm). Sahni~\cite{Sahni76} used \prob as one of the three first examples to illustrate the important concept of a fully polynomial time approximation scheme (FPTAS) in the area of scheduling. To that effect, several
generalizations of the \prob problem have been studied in the literature,
testing the limits to which these techniques can be applied~\cite{AdamuAdewumi2014}.

Another reason why \prob is such a prominent problem is that it is a natural generalization of two classical problems in combinatorial optimization. Indeed, the special case of \prob where all jobs have a common due date (\emph{i.e.} $d(x_1)= \cdots = d(x_n)=d$) translates directly to the dual version of \textsc{Knapsack}~\cite{Karp72}: In \prob our goal is to minimize the total weight of jobs that complete after $d$, where in \textsc{Knapsack} we wish to maximize the total weight of jobs that complete before $d$. (Here $d$ corresponds to the Knapsack size, the processing times correspond to item sizes, and the weights correspond to item values.) When in addition to $d(x_1)= \cdots = d(x_n)=d$, we also have $p(x)=w(x)$ for each job~$x$, the \prob problem becomes \textsc{Subset Sum}. The \probVariant problem, a generalization of \textsc{Subset Sum} and a special case of \prob, has recently received attention in the research community as well~\cite{BringmannFHSW22,Klein0R23,SchieberS23}.

\subsection{Parameterized complexity of \boldmath{\prob}}

In this paper we focus on the \prob problem from the perspective of parameterized complexity~\cite{DowneyFellows99,DowneyFeloows13}. Thus, we are interested to know whether there exists some algorithm solving \prob in $f(k) \cdot n^{O(1)}$ time, for some computable function~$f()$ and some problem-specific parameter~$k$. In parameterized complexity terminology this equates to asking whether \prob is \emph{fixed-parameter tractable} with respect to parameter~$k$. If we take $k$ to be the total weight of tardy jobs in an optimal schedule, then \prob is trivially fixed-parameter tractable by using the aforementioned pseudo-polynomial time algorithms that exist for the problem. In fact, these pseudo-polynomial time algorithms show that the \prob is only hard in the \emph{unbounded setting}, \emph{i.e.} the case where the processing times, weights, and due dates of the jobs may be super-polynomial in the number~$n$ of jobs. This is the case we focus on throughout the paper. 

In the unbounded setting, the most natural first step is to analyze \prob through the ``number of different numbers" lens suggested by Fellows \emph{et al.}~\cite{FellowsGR12}. In this framework, one considers problem instances with a small variety of numbers in their input. Three natural parameters arise in the context of the \prob problem: The number of different due dates~$d_{\#}=|\{d(x_1),\ldots,d(x_n)\}|$, the number of different processing times~$p_{\#}=|\{p(x_1),\ldots,p(x_n)\}|$, and the number of different weights~$w_{\#}=|\{w(x_1),\ldots,w(x_n)\}|$. Regarding parameter~$d_{\#}$, the situation is rather clear. Since \prob is essentially equivalent to the NP-hard \textsc{Knapsack} problem already for $d_{\#}=1$~\cite{Karp72}, there is no $f(k) \cdot n^{O(1)}$ time algorithm for the problem unless P=NP. 
\begin{theorem}[\cite{Karp72}]
\label{thm:karp}%
\prob is not fixed-parameter tractable with respect to~$d_{\#}$ unless P=NP.
\end{theorem}

What about parameters~$p_{\#}$ and~$w_{\#}$? This question was first studied in~\cite{HermelinKPS21}. There it was shown the \prob is polynomial time solvable when either~$p_{\#}$ or~$w_{\#}$ are bounded by a constant. This is done by generalizing the algorithms of Moore~\cite{Moore68} and Peha~\cite{Peha95} for the cases of $w_{\#}=1$ or $p_{\#}=1$. Moreover, the authors in~\cite{HermelinKPS21} show that any instance of \prob can be translates to an integer linear program whose number of variables depends solely on $d_{\#}+p_{\#}$, $d_{\#}+w_{\#}$, or~$p_{\#}+w_{\#}$. Thus, using fast integer linear program solvers such as Lenstra's celebrated algorithm~\cite{lenstra1983integer}, they proved that \prob is fixed parameter tractable with respect to all possible combinations of parameters~$d_{\#}$, $p_{\#}$, and $w_{\#}$.
\begin{theorem}[\cite{HermelinKPS21}]
\label{thm:known}%
The \prob problem is solvable in polynomial-time when~$p_{\#}=O(1)$ or~$w_{\#}=O(1)$. Moreover, it is fixed-parameter tractable with respect to parameters~$d_{\#}+p_{\#}$, $d_{\#}+w_{\#}$, and~$p_{\#}+w_{\#}$.
\end{theorem}

Thus, both the aforementioned \textsc{Knapsack} and \textsc{Subset Sum} problems are both fixed-parameter tractable in the number of different numbers viewpoint. What about \prob? The parameterized complexity status of \prob parameterized by either~$p_{\#}$ or~$w_{\#}$ was left open in~\cite{HermelinKPS21}, and due to Theorem~\ref{thm:karp} and Theorem~\ref{thm:known}, these are the only two remaining cases. Thus, the main open problem in this context is 
\begin{quote}
``Is \prob fixed-parameter tractable with respect to either $p_{\#}$ or $w_{\#}$?"
\end{quote}

\subsection{Our contribution}

In this paper we resolve the open question above negatively, by showing that \prob is W[1]-hard with respect to either~$p_{\#}$ or $w_{\#}$. This means that unless the central hypothesis of parameterized complexity is false, \prob is neither fixed-parameter tractable with respect to~$p_{\#}$ nor with respect to~$w_{\#}$.
\begin{theorem}
\label{thm:main}%
\prob parameterized by either $p_{\#}$ or $w_{\#}$ is W[1]-hard.
\end{theorem}
\noindent Thus, Theorem~\ref{thm:main} together with Theorem~\ref{thm:karp} and Theorem~\ref{thm:known} provide a complete picture of the parameterized complexity landscape of \prob with respect to parameters~$\{p_{\#},w_{\#},d_{\#}\}$, and any of their combinations.

We prove Theorem~\ref{thm:main} using an elaborate application of the ``multicolored clique technique''~\cite{FellowsHRV09} which we discuss later on. The proof gives one of the first examples of a single machine scheduling problem which is hard by the number of different processing times or weights. The only other example we are aware of is in~\cite{HermelinMO} for a generalization of \prob involving release times and batches. 
Indeed, there are several open problems regarding the hardness of scheduling problems with a small number of different processing times or weights. The most notable example is arguably the \minmakespanhighmultiplicity problem, whose parameterized complexity status is open for parameter~$p_{\#}$ (despite the famous polynomial-time algorithm for the case of~$p_{\#}=O(1)$~\cite{GoemansR20}). Further, Mnich and van Bevern~\cite{MnichOpenProblems} list three scheduling with preemption problems that are also open for parameter~$p_{\#}$. We believe that ideas and techniques used in our proof can prove to be useful for some of these problems as well. 

Regarding exact complexity bounds for \prob, the best known algorithms for the problem with respect to $p_{\#}$ and $w_{\#}$ have running times of the form $O(n^{k+1})$ for either $k=p_{\#}$ or $k=w_{\#}$~\cite{HermelinKPS21}. How much can we improve on these algorithms? A slight adaptation of our proof which we discuss in the last part of paper gives an almost complete answer to this question. In particular, we can show that the above upper bounds are tight up to a factor of $O(\lg k)$, assuming the Exponential Time Hypotheses (ETH) of Impagliazzo and Paturi~\cite{IP2001}.
\begin{restatable}[]{corollary}{ethhard}
\label{cor:eth-hardness}
\prob cannot be solved in $n^{o(k/\lg k)}$ time, for either $k=p_{\#}$ or $k=w_{\#}$, unless ETH is false.
\end{restatable}

\subsection{Technical overview}

We next give a brief overview of the proof of Theorem~\ref{thm:main}. As the case of parameter~$p_{\#}$ and~$w_{\#}$ are rather similar, let us focus on parameter~$p_{\#}$. On a high level, our proof follows the standard ``multicolored clique technique" introduced in~\cite{FellowsHRV09}. In this framework, one designs a parameterized reduction (see Definition~\ref{def:ParamReduction}) from  \mcc, where we are given a $k$-partite graph~$G=(V_1 \uplus \cdots \uplus V_k,E)$, and we wish to determine whether~$G$ contains a clique that includes one vertex from each \emph{color class}~$V_i$ of $G$ (see Figure~\ref{fig:examplegraph}). Given an instance of \mcc, our goal is to construct in $f(k) \cdot n^{O(1)}$ time an equivalent instance of \prob such that $p_{\#}=g(k)$ for some computable functions $f()$ and $g()$.

\begin{figure}
    \centering
    \begin{tikzpicture}
        \tikzstyle{vertex}=[draw, fill, inner sep= 3pt, circle]
        \tikzstyle{squared-vertex}=[draw, fill, inner sep= 3pt]

        \node[squared-vertex, label=180:$v_1^1$] (v11) at (0,0) {};
        \node[vertex, label=180:$v_2^1$] (v12) at (0.5,-0.5) {};
        \node[vertex, label=180:$v_3^1$] (v13) at (1,-1) {};
        \node[vertex, label=180:$v_4^1$] (v14) at (1.5,-1.5) {};
        \draw[rotate=-45] ($0.5*(v12) + 0.5*(v13)+(0, -0.225)$) ellipse (2cm and 0.6cm);
        \node at  ($0.5*(v12) + 0.5*(v13)+(-0.9, -0.8)$) {\Large $V_1$};

        \begin{scope}[xshift=0.5cm]
        \node[vertex, label=0:$v_1^2$] (v21) at (5.5,0) {};
        \node[squared-vertex, label=0:$v_2^2$] (v22) at (5,-0.5) {};
        \node[vertex, label=0:$v_3^2$] (v23) at (4.5,-1) {};
        \node[vertex, label=0:$v_4^2$] (v24) at (4,-1.5) {};
        \draw[rotate=45] ($0.5*(v22) + 0.5*(v23)+(0, -0.225)$) ellipse (2cm and 0.6cm);
        \node at  ($0.5*(v22) + 0.5*(v23)+(0.9, -0.8)$) {\Large $V_2$};
        \end{scope}

        \begin{scope}[xshift=0.25cm]
        \node[vertex, label=90:$v_1^3$] (v31) at (1.7,2) {};
        \node[vertex, label=90:$v_2^3$] (v32) at (2.4,2) {};
        \node[squared-vertex, label=90:$v_3^3$] (v33) at (3.1,2) {};
        \node[vertex, label=90:$v_4^3$] (v34) at (3.8,2) {};
        \node (V3) at ($0.5*(v32) + 0.5*(v33) + (0, 1.2)$) {\Large $V_3$};
        \draw ($0.5*(v32) + 0.5*(v33)+(0, 0.28)$) ellipse (2cm and 0.6cm);
        \end{scope}

        \draw (v11) edge[line width =2pt,dashed] (v33);
        \draw (v11) edge[dashed] (v34);
        \draw (v12) edge[dashed] (v31);
        \draw (v13) edge[dashed] (v33);
        
        \draw (v11) edge[line width =2pt] (v21);
        \draw (v11) edge[dashed, line width =3pt] (v22);
        \draw (v13) edge[dashed] (v23);
        \draw (v14) edge[dashed] (v24);
        
        \draw (v21) edge[line width =2pt] (v34);
        \draw (v21) edge[line width =2pt] (v33);
        \draw (v22) edge[dashed, line width =2pt] (v33);
        \draw (v23) edge[dashed] (v32);
    \end{tikzpicture}
    \caption{An example nice 3-partite graph with $n=4$ (the size of each color class) and $m=4$ (the number of edges between any pair of color classes). The selected vertices are squared. Lexicographically larger or equal edges are dashed, while smaller or equal edges are in bold.}
    \label{fig:examplegraph}
\end{figure}
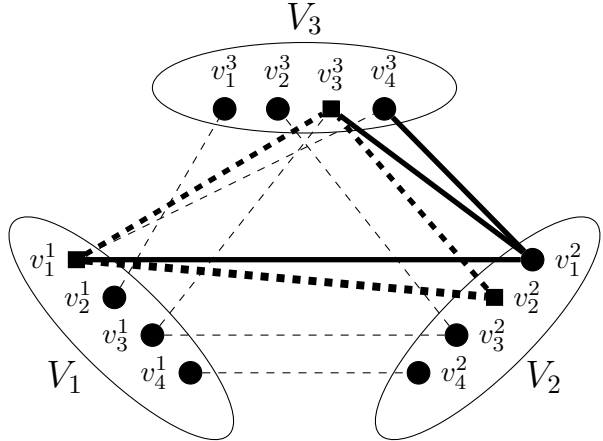

Our reduction essentially consists of three gadgets: One gadget for the vertices of $G$, and two gadgets for the edges of $G$. The gadget for the vertices of $G$, which we refer to as the \emph{vertex selection gadget}, consists of a set of jobs whose role is to encode the selection of a single vertex from each color class of~$G$. Since we may assume that each color class of~$G$ includes exactly $n$ vertices, we essentially need to encode the selection of $k$ integers $n_1,\ldots,n_k \in \{1,\ldots,n\}$. The crux is that we need to do this using jobs that have only~$f(k)$ different processing times.

The first edge gadget, called the \emph{large edge gadget}, consists of a set of jobs whose role is to count the number of edges that are lexicographically larger or equal to any selected pair $(n_i,n_j) \in \{1,\ldots,n\}^2$. The second edge gadget, referred to as the \emph{small edge gadget}, counts all lexicographically smaller or equal edges. In this way, if the total number of edges counted is $|E|+\binom{k}{2}$, then we know that the vertices indexed by $n_1,\ldots,n_k \in \{1,\ldots,n\}$ form a clique in~$G$. If the total number of counted edges is smaller, then $G$ contains no clique with $k$ vertices. Again, we need to ensure that the jobs in both gadgets have $f(k)$ different processing times.

To ensure all jobs constructed have a small variety of different processing times, we make heavy use of the fact that the processing times (and weights and due dates) can be rather large. Thus, we choose some polynomially large $N$, and use integers in the range of~$\{0,\ldots,N^{f(k)}-1\}$ for some function~$f()$. In this way, considering all integers in their base~$N$ representation, allows us to use the different \emph{digits} in the representation to encode various numerical values such as the integers $n_1,\ldots,n_k \in \{1,\ldots,n\}$. We partition each integer in~$\{0,\ldots,N^{f(k)}-1\}$ into \emph{blocks} of $m+2$ consecutive digits. Each digit in each block has a function that will overall allow us to use the strategy discussed above, and selecting a sufficiently large $N$ ensures that no overflow can occur between adjacent digits. The devil, of course, is in the details. 

\subsection{Roadmap}

The rest of the paper is organized as follows. In Section~\ref{sec:preliminaries} we briefly review all preliminary results that are necessary for proving our main result, \emph{i.e.} Theorem~\ref{thm:main}. Section~\ref{sec:PSharp} then contains the proof of Theorem~\ref{thm:main} for parameter~$p_{\#}$, which is the main technical part of the part. In Section~\ref{sec:WSharp} we discuss how to adapt the proof of Section~\ref{sec:PSharp} to parameter~$p_{\#}$. Finally, we discuss our ETH-based lower bounds in Section~\ref{sec:ETH}.

\section{Preliminaries}
\label{sec:preliminaries}

Throughout the paper we will use $<$ to denote the lexicographical order between ordered pairs of integers. Thus, 
$$
(i,j)<(i_0,j_0) \quad\iff\quad (i < i_0) \text{ or }  (i=i_0 \text{ and } j < j_0) 
$$
for any pair of integers $(i,j)$ and $(i_0,j_0)$.

\subsection{Parameterized complexity primer}

In parameterized complexity, an instance of a \emph{parameterized problem} $\Psi$ is a pair $(x,k) \in \{0,1\}^* \times \mathbb{N}$, where $x$ encodes the ``combinatorial part'' of the input (\emph{e.g.} a graph, a set of integers, ...), and $k$ is a numerical value representing the \emph{parameter}. Thus, when $\Psi$ is the \prob problem parameterized by $p_{\#}$, the string~$x$ encodes all processing times, weights, and due dates of the jobs, and $k$ equals the total number of different processing times in the input. The main tool we use for proving Theorem~\ref{thm:main} is that of a parameterized reduction:
\begin{definition}[\cite{DowneyFellows99}]
\label{def:ParamReduction}
A \emph{parameterized reduction} from a parameterized problem $\Psi_1$ to a parameterized problem $\Psi_2$ is an algorithm that receives as input an instance $(x,k)$ of $\Psi_1$ and outputs in $f(k) \cdot |x|^{O(1)}$ time for some computable function~$f()$ an instance $(y,\ell)$ of $\Psi_2$ such that
\begin{itemize}
\item $(x,k)$ is a yes-instance of $\Psi_1$ iff $(y,\ell)$ is a yes-instance of $\Psi_2$.
\item $\ell \leq g(k)$ for some computable function $g()$.
\end{itemize}
\end{definition}

A parameterized problem is said to be \emph{fixed-parameter tractable} (and in the class FPT) if there is an algorithm solving it in $f(k) \cdot |x|^{O(1)}$ time. The main hardness class in parameterized complexity is W[1]. Thus, the main working assumption in parameterized complexity is that FPT$\neq$W[1]. A parameterized problem is \emph{W[1]-hard} if there is a parameterized reduction from any problem in W[1] to~$\Psi$. If $\Psi_1$ is parameterized problem which is known to be W[1]-hard, and there exists a parameterized reduction from $\Psi_1$ to another parameterized problem~$\Psi_2$, then~$\Psi_2$ is also W[1]-hard~\cite{DowneyFellows99}.

\subsection{The multicolored clique problem}

The source W[1]-hard problem in our parameterized reduction used for proving Theorem~\ref{thm:main} is the \mcc problem. 
\begin{definition}
Given a $k$-partite graph $G=(V_1 \uplus \cdots \uplus V_k,E)$, the \mcc\ problem asks to determine whether $G$ contains a subset of $k$ pairwise adjacent vertices (\emph{i.e.}, a clique of size $k$).
\end{definition}

For a given a $k$-partite graph $G=(V_1 \uplus \cdots \uplus V_k,E)$, we let $E_{i,j}$ denote the set of edges between any vertex in $V_i$ and any vertex in $V_j$, for all $1 \leq i < j \leq k$. We say that a $k$-partite graph $G=(V_1\cup \cdots \cup V_k)$ is \emph{nice} if $|V_1|= \cdots = |V_k|$ and $|E_{1,2}| = \cdots = |E_{k-1,k}|$. 
\begin{theorem}[\cite{FellowsHRV09}]
\label{thm:mcc}
\mcc\ is W[1]-hard when parameterized by~$k$, even if the input graph is nice.    
\end{theorem}

Given a nice $k$-partite graph $G=(V_1 \uplus \cdots \uplus V_k,E)$, we refer to each $V_i \in \{V_1,\ldots,V_k\}$ as a \emph{color class} of $G$. We write $V_i=\{v^i_1,\ldots,v^i_n\}$ to denote vertices in $V_i$ for each $1 \leq i \leq k$, and $E_{i,j}=\{e^{i,j}_1,\ldots,e^{i,j}_m\}$ to denote the edges in $E_{i,j}$ for each $1 \leq i < j \leq k$. When considering a specific set of edges $E_{i,j}$, we will often use $\ell_i \in \{1,\ldots,n\}$ and $\ell_j \in \{1,\ldots,n\}$ to respectively denote the index of the vertex in $V_i$ and the index of the vertex of $V_j$ in the $\ell$'th edge of $E_{i,j}$. That is, $e^{i,j}_\ell=\{v^i_{\ell_i},v^j_{\ell_j}\}$.

\subsection{EDD schedules}

In the \prob problem it is frequently convenient to work with what we refer to as an EDD\footnote{EDD here is an acronym for ``Earliest Due Date".} schedule.
\begin{definition}
A schedule $\Pi$ for a set $\{x_1,\ldots,x_n\}$ of jobs is \emph{EDD} if all early jobs in $\Pi$ are scheduled in before all tardy jobs, and the order among early jobs is non-decreasing in due dates. Thus, if $\Pi(x_i) < \Pi(x_j)$, then either $x_j$ is tardy, or both jobs are early and $d(x_i) \leq d(x_j)$.
\end{definition}

The reason EDD schedules are popular when working with the \prob problem is that we can always assume that there exists an optimal schedule which is EDD. The following lemma is by now folklore (see \emph{e.g.}~\cite{AdamuAdewumi2014}), and can easily be proven by an exchange argument which swaps early jobs that do not satisfy the EDD property in a given optimal schedule.
\begin{lemma}
\label{lem:EDD}%
Any instance of \prob has an optimal EDD schedule.     
\end{lemma}

Thus, throughout our reduction from \mcc, we can restrict our attention to EDD schedules only.
Given an EDD schedule $\Pi_0$ for a job set $\{x_1,\ldots,x_n\}$, we say that $\Pi$ is \emph{an extension of $\Pi_0$ to the set of jobs $\{y_1,\ldots,y_m\}$} if $\Pi$ is an EDD schedule for $\{x_1, \ldots, x_n, y_1,\ldots,y_m\}$ which schedules early all jobs that are scheduled early in $\Pi_0$. We write~$P(\Pi)$ and~$W(\Pi)$ to  respectively denote the total processing time and weight of all early jobs in a given EDD schedule~$\Pi$.

\section{Parameter \boldmath{$p_{\#}$}}
\label{sec:PSharp}%

In the following section we present a proof of Theorem~\ref{thm:main} for parameter~$p_{\#}$. As mentioned above, the proof consists of a parameterized reduction from \mcc\ parameterized by~$k$ to \prob parameterized by $p_{\#}$. We use $G=(V=V_1 \uplus \cdots \uplus V_k,E)$ to denote an arbitrary nice $k$-partite graph given as an instance of \mcc, with $n=|V_1|=\cdots=|V_k|$ and $m=|E_{1,2}|=\cdots=|E_{k-1,k}|$. Before discussing our construction in full detail, we review the terminology that we will use throughout for handling large integers.

\subsection{Digits and blocks}

Let $N$ be a polynomially-bounded integer that is chosen to be sufficiently larger than the overall number of jobs in our construction ($N=O(kn+k^2m)$ is enough). This number will appear frequently in the processing times, weights, and due dates of the jobs in our construction. In particular, it is convenient to view each integer in our construction in its base~$N$ representation: Each integer will be in the range of $[0,1,\ldots,N^{k+2\binom{k}{2}\cdot(m+2)+1}-1]$, and so we can view each integer as a string of length $k+2\binom{k}{2} \cdot (m+2)+1$ over the alphabet~$\{0,\ldots,N-1\}$. When viewed as such, we will refer to each letter of the string as a \emph{digit}. 

Furthermore, we will conceptually partition each integer into \emph{blocks} of consecutive digits as follows (see Example~\ref{ex:blocks}): The least most significant digit is a block within itself which we refer to as the \emph{counting block}. Following this, there are $\binom{k}{2}$ blocks which we refer to as the \emph{small blocks}, consisting of $m+2$ digits each, where the first (least significant) block corresponds to the color class pair $(V_1,V_2)$, the second corresponds to $(V_1,V_3)$, and so forth. Following the small blocks are $\binom{k}{2}$ blocks which we dub the \emph{large blocks}, which again consist of $m+2$ digits each, and are ordered similarly to the left blocks. The final block is the \emph{vertex selection block} which consists of the $k$ most significant digits of the given integer.

\begin{example}
\label{ex:blocks}
As example, the following is the partitioning of integer 0:
$$
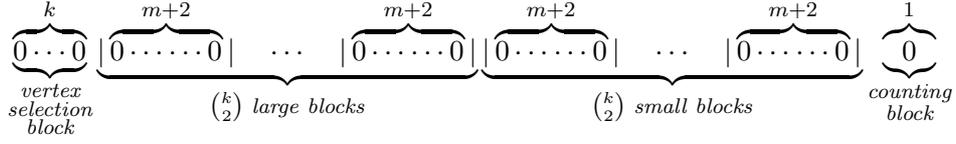
\underbrace{\overbrace{0 \cdots 0}^k}_{\substack{\text{vertex}\\ \text{selection}\\ \text{block}}} \underbrace{|\overbrace{0 \cdots \cdots 0}^{m+2}| \quad \cdots \quad | \overbrace{0 \cdots \cdots 0}^{m+2}|}_{\binom{k}{2} \text{ large blocks}} \underbrace{|\overbrace{0 \cdots \cdots 0}^{m+2}| \quad \cdots \quad | \overbrace{0 \cdots \cdots 0}^{m+2}|}_{\binom{k}{2}\text{ small blocks}} \underbrace{\overbrace{0}^1}_{\substack{\text{counting}\\\text{block}}}
$$
The large and small blocks are ordered in increasing lexicographic order of $(i,j)$, so the $(1,2)$ small block is the first block following the counting block. In each block we order the digits from least significant to most significant, so the first digit in the $(1,2)$ small block is the second least significant digit overall.
\end{example}

Let $g:\{(i,j) \mid 1 \le i < j \le k\} \to \{0,\ldots,\binom{k}{2}-1\}$ denote the lexicographic ordering function, that is $g(i,j) > g(i_0,j_0)$ iff $(i,j) > (i_0,j_0)$ for all $1 \le i < j \le k$. Furthermore, let $G(i,j)=(m+2)\cdot g(i,j) +1$ for all $1 \leq i < j \leq k$. Similarly, let $f:\{(i,j) \mid 1 \le i < j \le k\} \to \{\binom{k}{2},\ldots,2\cdot\binom{k}{2}-1\}$ denote the function defined by $f(i,j)=\binom{k}{2}+g(i,j)$, and let $F(i,j)=(m+2)\cdot f(i,j) +1$. We will use the following constants in our construction:
\begin{itemize}
\item $X_i := N^{(m+2) \cdot 2\binom{k}{2} +i}$ for $i \in \{1,\ldots,k\}$,
\item $Y_{i,j} := N^{F(i,j)+m+1}$  for $i < j\in \{1,\ldots,k\}$, and
\item $Z_{i,j} := N^{G(i,j)+m+1}$ for $i < j\in \{1,\ldots,k\}$.
\end{itemize}
Thus, $X_i$ corresponds to the $i$'th digit in the vertex selection block, $Y_{i,j}$ corresponds to the last digit in the $(i,j)$ large block, and $Z_{i,j}$ corresponds to the last digit in the $(i,j)$ small block.

\subsection{Vertex selection gadget}
\label{sec:vertex-selection}

The role of the vertex selection gadget is to encode the selection of $k$ vertices, one from each color class $V_i$ of $G$. In constructing the vertex selection jobs, we will use the following two values associated with each $i\in\{1,\ldots,k\}$:
\begin{itemize}
\item $L(i) = \sum_{j= 1}^{i-1} N^{F(j,i)} + \sum_{j = i+1}^k N^{F(i,j) +1}$.
\item $S(i) = \sum_{j = 1}^{i-1} N^{G(j,i)} + \sum_{j = i+1}^k N^{G(i,j) +1}$.
\end{itemize}
Thus, adding $L(i)$ to an integer corresponds to adding a 1 to the first digit of every $(j,i)$ large block with $j <i$, and a 1 to the second digit of any $(i,j)$ large block with $j > i$. Adding~$S(i)$ corresponds to adding a 1 to the same digits in the small blocks.

Let $1 \leq i \leq k$, and consider the color class~$V_i$ of~$G$. The $V_i$ vertex selection gadget is constructed as follows. Let $P^V_i$ denote the following value:
$$
P^V_i \,\,=\,\, n \cdot \sum _{j > i} X_j \,\,=\,\, n \cdot \sum _{j > i} N^{(m+2) \cdot 2\binom{k}{2} +j}.
$$
Thus, $P^V_i$ has $n$ as its $j$'th most significant digit for $j < i$, and 0 in all of its other digits. We construct $n-1$ copies of the job pair $\{x_i,\neg x_i\}$ with the following characteristics:
\begin{itemize}
\item $p(x_i)=w(x_i)=X_i + L(i)$. 
\item $p(\neg x_i)=w(\neg x_i)=X_i + S(i)$.
\item $d(x_i)  = d(\neg x_i) = P^V_{i-1} + N^{(m+2) \cdot 2\binom{k}{2} }$ (where $P^V_0=n \cdot \sum_i X_i$).
\end{itemize}

In addition to these $n-1$ copies of~$\{x_i,\neg x_i\}$, we construct a single job $x^*_i$ with similar processing time and due date as $x_i$, but with significantly larger weight:
\begin{itemize}
\item $p(x^*_i)=p(x_i)$ and $d(x^*_i) = d(x_i)$.
\item $w(x^*_i)=(n+1) \cdot X_i + L(i)$.
\end{itemize}
The jobs~$x_i^*, x_i$, and $\neg x_i$ are called \emph{$V_i$ vertex selection jobs}.

Overall, we have $(2n-1) \cdot k$ vertex selection jobs that together have only $2k$ different processing times, $3k$ different weights, and $k$ different due dates. The vertex selection jobs are constructed in a way so that any schedule with sufficiently large weight of early jobs will schedule $n$ early jobs from $\{x^*_i,x_i,\neg x_i\}$ for each $1 \leq i \leq n$. Due to the large weight of $x^*_i$, job~$x_i^*$ will always be scheduled early, while the number of early jobs $x_i$ and $\neg x_i$ will be used to encode an integer $n_i \in \{1,\ldots,n\}$ corresponding to vertex $v^i_{n_i} \in V_i$. 

\begin{lemma}
\label{lem:VertexSelectForward}
Let $n_1,\ldots,n_k \in \{1,\ldots,n\}$. There exists a schedule $\Pi=\Pi(n_1,\ldots,n_k)$ for the vertex selection jobs such that for each $i \in \{1,\ldots,k\}$ precisely $n_i$ jobs from $\{x^*_i,x_i\}$ and $n-n_i$ copies of $\neg x_i$ are early in $\Pi$ for each $1 \leq i \leq k$     
\end{lemma}

\begin{proof}
For $i=k,\ldots,1$, we proceed as follows: We schedule job $x^*_i$, followed by $n_i-1$ copies of~$x_i$ and $(n-n_i)$ copies of $\neg x_i$. By construction, the total processing time of all scheduled jobs from the $V_i$ vertex selection gadget is 
$$
n \cdot X_i + n_i \cdot L(i) + (n-n_i) \cdot S(i), 
$$
and since
\begin{align*}
\sum _{j \ge i} \left(n \cdot X_j +  n_j \cdot L(j) +  (n-n_j) \cdot S(j) \right)\,\,&=\\ 
P_{i-1}^V + \sum_{j \ge i} n_j \cdot L(j) + \sum_{j \ge i} (n-n_j) \cdot S(j) \,\,&\leq\\ 
P_{i-1}^V + N^{(m+2) \cdot 2\binom{k}{2}} \,\, &= \,\,d(x^*_i) \,\,= \,\,d(x_i)\,\,=\,\,d(\neg x_i),
\end{align*}
all scheduled jobs are early.
\end{proof}

Throughout the remainder of the proof we will use $\Pi=\Pi(n_1,\ldots,n_k)$ to denote the schedule that schedules $x_i^*$, exactly $n_i-1$ jobs~$x_i$, and $n-n_i$ jobs $\neg x_i$ early. Let $W_V$ denote the value
$$
W_V = 2n \cdot \sum_i X_j. 
$$
Then the following corollary follows directly from Lemma~\ref{lem:VertexSelectForward}:
\begin{corollary}
\label{cor:VertexSelection}
Let $\Pi=\Pi(n_1,\ldots,n_k)$ for some $n_1,\ldots,n_k \in \{1,\ldots,n\}$. Then 
\begin{itemize}
\item[$(i)$] $P(\Pi)=P_V+ \sum_i n_i \cdot L(i) + \sum_i (n-n_i) \cdot S(i)$. 
\item[$(ii)$] $W(\Pi)=W_V+ \sum_i n_i \cdot L(i) + \sum_i (n-n_i) \cdot S(i)$.
\end{itemize} 
\end{corollary}

\begin{example}
\label{ex:VertexSelection}%
Consider the 3-partite graph in Figure~\ref{fig:examplegraph}, where vertex $v^i_i$ is selected for each color class $V_i$. Then the total processing time of all early vertex selection jobs in this example is:
$$
444 | \underbrace{0000 23|}_{\substack{(2,3)\\ \text{large}}} | \underbrace{ 0000 13|}_{\substack{(1,3)\\ \text{large}}} | \underbrace{0000 12|}_{\substack{(1,2)\\ \text{large}}} | \underbrace{0000 21|}_{\substack{(2,3)\\ \text{small}}} | \underbrace{0000 31|}_{\substack{(1,3)\\ \text{small}}} | \underbrace{0000 32|}_{\substack{(1,2)\\ \text{small}}} | 0
$$
The total weight of all early vertex selection jobs is identical, except that the vertex selection block equals `$888$' instead of `$444$'.
\end{example}

As mentioned above, the vertex selection jobs are constructed in a way so that any schedule~$\Pi$ for these jobs with sufficiently large weight of early jobs will schedule precisely $n$ early jobs from~$\{x^*_i,x_i,\neg x_i\}$ for each $1 \leq i \leq k$. This is formally proven in the following lemma:

\begin{lemma}
\label{lem:VertexSelectBackward}%
Let $\Pi$ be an EDD schedule for the vertex selection jobs with $W(\Pi) \geq W_V$. Then $\Pi=\Pi(n_1,\ldots,n_k)$ for some $n_1,\ldots,n_k \in \{1,\ldots,n\}$.
\end{lemma}

\begin{proof}
Let $\Pi_i$ denote the restriction of $\Pi$ to the $V_i$ vertex selection jobs for~$j >i$. To prove the lemma, we prove the following stronger statement by backward induction on~$i$: If $w(\Pi_i) \geq W^V_{i-1}$ where $W^V_i := 2n \cdot \sum_{j> i} X_i$, then $\Pi_i=\Pi(n_i,\ldots,n_k)$ for some $n_1,\ldots,n_k \in \{1,\ldots,n\}$. Equivalently, we show for each $i= k,\ldots,1$ that if $w(\Pi_i) \geq W^V_{i-1}$ then $\Pi_i$ schedules $x^*_i$ and precisely $n-1$ jobs from the job pair~$\{x_i,\neg x_i\}$.

Let $i=k$, and suppose $w(\Pi_k) \geq W^V_{k-1}$. First note that since $N^{(m+2) \cdot 2\binom{k}{2}} < X_k$, we have $d(x^*_k) = d(x_k) = d(\neg x_k) < (n + 1) \cdot X_k$, and so at most $n$ jobs from~$\{x^*_k,x_k,\neg x_k\}$ can be scheduled early in $\Pi_k$. Since the total weight of any set of $n$ vertex selection jobs that does not include $x^*_k$ is less than $(n+1) \cdot  X_k$, it must be that $x^*_k$ is scheduled early in $\Pi_k$ as otherwise $W(\Pi_k) < (n+1) \cdot \sum X_k < W^V_{k-1}$. Moreover, as the total weight of all vertex selection jobs outside the $V_k$ vertex selection job is less than $X_k$, it must be that $n-1$ jobs from $\{x_k,\neg x_k\}$ are scheduled early in $\Pi$, as otherwise $W(\Pi_k) < 2n \cdot X_k = W^V_{k-1}$.

Now let $1 \leq i < k$. By induction, we have $\Pi_{i+1}=\Pi(n_k,\ldots,n_{i+1})$ for some $n_{i+1},\ldots,n_k \in \{1,\ldots,n\}$. We refer to the jobs jobs of any $V_j$ vertex selection gadget with~$j \geq i$ as the \emph{remaining jobs}. As $\Pi_i$ is an EDD schedule, it schedules jobs from the $V_i$ vertex selection gadget starting at time~$P(\Pi_{i+1})$. Since $P(\Pi_{i+1}) > P^V_i$ by construction, and as $X_i > N^{(m+2) \cdot 2\binom{k}{2}}$, we have
$$
P(\Pi_{i+1}) + (n+1) \cdot X_i \,\,>\,\, P^V_i + (n+1) \cdot X_i \,\,>\,\, d(x^*_i) = d(x_i) = d(\neg x_i).
$$
Thus, at most $n$ jobs from~$\{x^*_i,x_i,\neg x_i\}$ are scheduled early in $\Pi_i$. Moreover, since the total weight of any set of $n$ remaining jobs that does not include $x^*_i$ is less than $(n+1) \cdot \sum_i X_k$, it must be that $x^*_i$ is scheduled early in $\Pi_i$ as otherwise $W(\Pi_i) < (n+1) \cdot X_i + 2n \cdot \sum_{j>i} X_j < W^V_{j-1}$. Moreover, as the total weight of all remaining jobs outside the $V_i$ vertex selection gadget is less than $X_i$, it must be that $n-1$ jobs from $\{x_i,\neg x_i\}$ are scheduled early in $\Pi$, as otherwise $W(\Pi_i) < 2n \cdot \sum_{j\ge i} X_j = W^V_{i-1}$. 
\end{proof}

\subsection{Large edge gadget}
\label{sec:large-edge}

We next describe the large edge gadget. The role of this gadget is to ``count'' all edges that are lexicographically larger or equal to pairs of selected vertices. This is done by constructing a pair of jobs $\{y^{i,j}_\ell, \neg y^{i,j}_\ell\}$ for each edge $e^{i,j}_\ell$ of $G$, along with some additional filler jobs.

Let $1 \leq i <j \leq k$. The $(i,j)$ large edge gadget is constructed as follows. First we define~$P^L_{i,j}$ to be the following value: 
$$
P^L_{i,j} \,\,\, =\,\,\, \sum_{(i_0,j_0) > (i,j)} \left(m \cdot Y_{i_0,j_0} + n \cdot N^{F(i_0,j_0)+1} +  n \cdot N^{F(i_0,j_0)} \right).   
$$
Thus, the two first digits of the $(i,j)$ large block in $P^L_{i,j}$ equal $n$, the last digit of this block equals $m$, and all other digits equal 0. Let $\ell \in \{1,\ldots,m\}$, and suppose that the $\ell$'th edge between~$V_i$ and~$V_j$ is the edge $e^{i,j}_\ell =\{v^i_{\ell_i},v^j_{\ell_j}\}$ for some $\ell_i,\ell_j \in \{1,\ldots,n\}$. We construct two jobs~$y^{i,j}_\ell$ and~$\neg y^{i,j}_\ell$ corresponding to~$e^{i,j}_\ell$ with the following characteristics:
\begin{itemize}
\item $p(y^{i,j}_\ell)=Y_{i,j}$ and $w(y^{i,j}_\ell) = Y_{i,j} \,/\, N^\ell + 1$.
\item $p(\neg y^{i,j}_\ell)=Y_{i,j}$ and $w(\neg y^{i,j}_\ell) = Y_{i,j} \,/\, N^\ell$.
\item $d(y^{i,j}_\ell)= P_V + P^L_{i,j} +  \ell \cdot Y_{i,j} + \ell_i \cdot N^{F(i,j)+1} + \ell_j \cdot N^{F(i,j)} + N^{F(i,j)-1}$.
\item $d(\neg y^{i,j}_\ell)= P_V + P^L_{i,j} +  \ell \cdot Y_{i,j} + n \cdot N^{F(i,j)+1} + n \cdot N^{F(i,j)} + N^{F(i,j)-1}$.
\end{itemize}
Observe that both jobs have the same processing time, which is equal throughout for jobs corresponding to other edges of $E_{i.j}$. Also note that the weight of~$y^{i,j}_\ell$ is slightly larger than the weight of~$\neg y^{i,j}_\ell$, while the due date of of~$\neg y^{i,j}_\ell$ is significantly larger than the due date of $y^{i,j}_\ell$. 

We will also need to add \emph{filler jobs} that will help us control the total processing times of all early jobs selected from the $(i,j)$ large edge gadget. We construct $n$ copies of the the job pair~$\{f^{i,j}_0,f^{i,j}_1\}$ which have the following characteristics:
\begin{itemize}
\item $p(f^{i,j}_0)=w(f^{i,j}_0)=N^{F(i,j)}$.
\item $p(f^{i,j}_1)=w(f^{i,j}_1)=N^{F(i,j)+1}$.
\item $d(f^{i,j}_0) = d(f^{i,j}_1) = P_V + P^L_{i,j} + m \cdot Y_{i,j} + n \cdot N^{F(i,j)+1} + n \cdot N^{F(i,j)} + N^{F(i,j)-1}$.
\end{itemize}

Thus, altogether, the large edge gadget consists of the job pair $\{y^{i,j}_\ell,\neg y^{i,j}_\ell\}$ for $\ell \in \{1,\ldots,m\}$ and~$n$ copies of the job pair $\{f^{i,j}_0,f^{i,j}_1\}$, for each $1 \leq i < j \leq k$. Note that the large edge jobs have $3\binom{k}{2}$ different processing times in total. We next prove a lemma regarding the structure of certain schedules for the vertex selection and large edge job. This structure is what allows us to count all edges that are lexicographically larger or equal any selected pair~$(n_i,n_j)$. Let~$\Pi_V$ be a schedule for the vertex selection jobs. We say that $\Pi$ is an \emph{optimal extension} of $\Pi_V$ to the set of large edge jobs if all jobs that are early in $\Pi_V$ are also early in $\Pi$, and there is no other such schedule with a larger total weight of early jobs. 

\begin{restatable}[]{lemma}{lemLargeEdge}
\label{lem:LargeEdge}%
Let~$\Pi_V=\Pi (n_1,\ldots,n_k)$ be a schedule for the vertex selection jobs for some $n_1,\ldots,n_k \in \{1,\ldots,n\}$, and let $\Pi$ be an optimal extension of $\Pi_V$ to the set of large edge jobs. Then the following properties hold for each~$1 \leq i < j \leq k$: 
\begin{itemize}
\item[(a)] The total processing time $P$ of all vertex selection jobs and all $(i_0,j_0)$ large jobs for $(i_0,j_0) > (i,j)$ which are early in $\Pi$ satisfies
$$
P \,\, \geq \,\, P_V + P^L_{i,j} + n_i \cdot N^{F(i,j)+1} + n_j \cdot N^{F(i,j)} 
$$
and 
$$
P \,\, \leq \,\, P_V + P^L_{i,j} + n_i \cdot N^{F(i,j)+1} + n_j \cdot N^{F(i,j)} + N^{F(i,j)-1}.
$$
\item[(b)] For each $\ell \in \{1,\ldots,m\}$ we have that either job $y^{i,j}_\ell$ or job $\neg y^{i,j}_\ell$ is early in $\Pi$, but not both. Job~$y^{i,j}_\ell$ is early iff $(n_i,n_j) \leq (\ell_i,\ell_j)$, 
where $e^{i,j}_\ell=\{v^i_{\ell_i},v^j_{\ell_j}\}$ is the $\ell$'th edge in~$E_{i,j}$.
\item[(c)] Precisely $n- n_i$ copies of job $f^{i,j}_1$ and $n- n_j$ copies of job $f^{i,j}_0$ are scheduled early in~$\Pi$.
\end{itemize}
\end{restatable}

\begin{proof}
We prove that~$\Pi$ satisfies the properties of lemma by backward induction on $(i,j)$, starting with the base case of $(i,j)=(k-1,k)$. 

Consider first property $(a)$: Observe that $P= P(\Pi_V)$ in this case, and that $P^L_{k-1,k}=0$. Now, according to Corollary~\ref{cor:VertexSelection}, we have that
\begin{align*}
P = P(\Pi_V) \,&= \,  P_V + \sum_i n_i \cdot L(i) + \sum_i (n-n_i) \cdot S(i)\\
&\geq\, P_V + P^L_{k-1,k} + n_{k-1} \cdot N^{F(k-1,k)+1} + n_k \cdot N^{F(k-1,k)}. 
\end{align*}
On the other hand, as $P(\Pi_0)$ is maximized when $n_1,\ldots,n_{k-2}=n$, we have  
\begin{align*}
P \,\,&\leq\,\,  P_V + n_{k-1} \cdot L(k-1) + (n-n_{k-1}) \cdot S(k-1)+ n_k \cdot L(k) + (n-n_k) \cdot S(k) + n \cdot \sum^{k-2}_{i=1} L(i) \\
&\leq\,\, P_V + n_{k-1} \cdot N^{F(k-1,k)+1} + n_k \cdot N^{F(k-1,k)} + N^{F(k-1,k)-1}.
\end{align*}
Thus, property $(a)$ holds for $(i,j)=(k-1,k)$.

We next prove property $(b)$ for $(i,j)=(k-1,k)$ by induction on $\ell$. Let $\ell =1$. Note that as both $w(y^{k-1,k}_1)$ and~$w(\neg y^{k-1,k}_1) \geq N^{F(k-1,k)+m+1}$ are larger than the total weight of all other large edge jobs, one of these jobs must be early in~$\Pi$, as otherwise $\Pi$ is not an optimal extension of $\Pi_V$. Moreover, as $\Pi$ is an EDD schedule, it schedules either job~$y^{k-1,k}_1$ or job~$\neg y^{k-1,k}_1$ at time~$P$. By property~$(a)$ we have that
\begin{align*}
P+p(y^{k-1,k}_1) + p(\neg y^{k-1,k}_1) \,\, &> \,\, P_V + p(y^{k-1,k}_1) + p(\neg y^{k-1,k}_1) \\
&=\,\, P_V +2 \cdot Y_{k-1,k}\\ 
& > \,\, P_V + Y_{k-1,k} + n \cdot N^{F(k-1,k)+1} + n \cdot N^{F(k-1,k)} + N^{F(k-1,k)-1}\\ 
&=\,\, d (\neg y^{k-1,k}_1) \,\, > \,\, d(y^{k-1,k}_1), 
\end{align*}
and so at most one of $y^{k-1,k}_1$ and~$\neg y^{k-1,k}_1$ can be early. On the other hand, we have 
\begin{align*}
P + p(\neg y^{k-1,k}_1) \,\,&\leq \,\,  P_V + n_{k-1} \cdot N^{F(k-1,k)+1} + n_k \cdot N^{F(k-1,k)} + N^{F(k-1,k)-1} + Y_{k-1,k} \\
&\leq\,\, P_V + n \cdot N^{F(k-1,k)+1} + n \cdot N^{F(k-1,k)} + N^{F(k-1,k)-1} + Y_{k-1,k}  \,\, = \,\, d(\neg y^{k-1,k}_1),
\end{align*}
and so at least $\neg y^{k-1,k}_1$ can be scheduled early.

However, as $w(y^{k-1,k}_1) > w(\neg y^{k-1,k}_1)$ and $p(y^{k-1, k}_1) = p(\neg y^{k-1, k}_1)$, an optimal extension of~$\Pi_V$ would schedule job~$y^{k-1,k}_1$ early if possible. If $(n_{k-1},n_k) \leq (\ell_{k-1},\ell_{k})$ then
\begin{align*}
P + p(y^{k-1,k}_1) \,\, &\leq \,\, P_V + n_{k-1} \cdot N^{F(k-1,k)+1} + n_k \cdot N^{F(k-1,k)} + N^{F(k-1,k)-1} + Y_{k-1,k} \\
&\leq \,\, P_V + \ell_{k-1} \cdot N^{F(k-1,k)+1} + \ell_k \cdot N^{F(k-1,k)} + N^{F(k-1,k)-1} + Y_{k-1,k}  = d(y^{k-1,k}_1),
\end{align*}
and so~$y^{k-1,k}_1$ is indeed early in~$\Pi$. If $(n_{k-1},n_k) > (\ell_{k-1},\ell_k)$, then
\begin{align*}
P + p(y^{k-1,k}_1) \,\, &\geq \,\, P_V + n_{k-1} \cdot N^{F(k-1,k)+1} + n_k \cdot N^{F(k-1,k)} +Y_{k-1,k} \\
&>\,\, P_V + n_{k-1} \cdot N^{F(k-1,k)+1} + (n_k-1) \cdot N^{F(k-1,k)} + N^{F(k-1,k)-1} +Y_{k-1,k}\\
&\geq\,\, P_V + \ell_{k-1} \cdot N^{F(k-1,k)+1} + \ell_k \cdot N^{F(k-1,k)} + N^{F(k-1,k)-1} + Y_{k-1,k}= d(y^{k-1,k}_1), 
\end{align*}
and so~$y^{k-1,k}_1$ is not early in~$\Pi$. Thus, property~$(b)$ holds for $\ell=1$. The inductive step for~$\ell >1$ follows by the exact same arguments while observing that the weight of either~$y^{i,j}_\ell$ or $\neg y^{i,j}_\ell$ is larger than the total weight of all remaining jobs (\emph{i.e.} all large edge jobs except jobs in $\{y^{k-1,k}_1,\neg y^{k-1,k}_1, \ldots, y^{k-1,k}_{\ell-1},\neg y^{k-1,k}_{\ell-1}\}$), and that the due date of both of these jobs contains the term $\ell \cdot Y_{i,j}$.

Finally, let us consider property~$(c)$. As~$\Pi$ is an EDD schedule, it first schedules all early vertex selection jobs, followed by all early $(k-1,k)$ large edge jobs, all early filler jobs of type~$f^{k-1,k}_1$, and then all early filler jobs of type~$f^{k-1,k}_0$. Due to property~$(b)$, the total processing time of all early jobs in~$\Pi$ prior to the filler jobs is~$P^*=P+m \cdot Y_{k-1,k}$. According to property~$(a)$, we can schedule $n-n_{k-1}$ copies of the filler job~$f^{k-1,k}_1$ since
\begin{align*}
P^*+(n-n_{k-1}) \cdot p(f^{k-1,k}_1) \,\, &\leq \,\, P_V \!+\! n \!\cdot\! N^{F(k-1,k)+1} \!+\! n_{k} \!\cdot\! N^{F(k-1,k)} \!+\! N^{F(k-1,k)-1} \!+\! m \!\cdot\! Y_{k-1,k} \\
&\leq\,\,  P_V \!+\! n \!\cdot\! N^{F(k-1,k)+1} \!+\! n \!\cdot\! N^{F(k-1,k)} \!+\! N^{F(k-1,k)-1} \!+\! m \!\cdot\! Y_{k-1,k}\\ 
&=\,\, d(f^{k-1,k}_1).    
\end{align*}
Scheduling more than $n-n_k$ copies is not possible since $P^* + (n-n_k+1) \cdot p(f^{k-1,k}_1)$ is at least 
$$
P_V + (n+1) \cdot N^{F(k-1,k)+1} + n_{k} \cdot N^{F(k-1,k)} +m \cdot Y_{k-1,k}> d(f^{k-1,k}_1). 
$$
It follows that~$\Pi$ schedules precisely $n-n_{k-1}$ copies of ~$f^{k-1,k}_1$. A similar argument shows that~$\Pi$ also schedules precisely $n-n_k$ copies of ~$f^{k-1,k}_0$.

We have thus shown that the lemma holds for the base case of~$(i,j)=(k-1,k)$. The inductive step for property~$(a)$ follows by observing that, by induction, the total processing time of all early $(i_0,j_0)$ large edge jobs in~$\Pi$ with $(i_0,j_0) > (i,j)$ is exactly 
$$
P^L_{i,j} \,\,\,- \sum_{(i_0, j_0) > (i,j)} \bigl(n_{i_0} \cdot N^{F(i_0,j_0) + 1} + n_{j_0} \cdot N^{F(i_0,j_0)}\bigr).
$$
Properties $(b)$ and~$(c)$ then follow using the same arguments as above.
\end{proof}

Using Lemma~\ref{lem:LargeEdge}, we can again derive the total processing time and weight of all early jobs in any optimal EDD schedule $\Pi$ for vertex selection jobs and the large edge jobs. Let~$W^L$ be the following value: 
$$
W_L = \sum_{(i,j)}  \left(\sum_\ell Y_{i,j} \,/\, N^{\ell} + n \cdot N^{F(i,j)+1} +  n \cdot N^{F(i,j)} \right). 
$$
Define $P_L=P^L_{0,0}$. Moreover, for~$1 \leq i < j \leq k$ and~$n_i,n_j \in \{1,\ldots,n\}$, define $m^L_{i,j}(n_i,n_j)$ to be the total number of edges in $E_{i,j}$ that are lexicographically larger or equal to~$(n_i,n_j)$. That is, the total number of edges $e^{i,j}_\ell=(v^i_{\ell_i},v^j_{\ell_j})\in E_{i,j}$ with $(n_i,n_j) \leq (\ell_i,\ell_j)$. Then the following holds:

\begin{restatable}[]{corollary}{LargeEdge}
\label{cor:LargeEdge}
Let~$\Pi_V=(n_1,\ldots,n_k)$ be a schedule for the vertex selection jobs for some $n_1,\ldots,n_k \in \{1,\ldots,n\}$, and let $\Pi$ be an optimal extension of $\Pi_V$ to the set of vertex selection and large edge jobs. Then
\begin{itemize}
\item[$(i)$] $P(\Pi)=P_V+P_L+\sum_i (n-n_i) \cdot S(i)$.
\item[$(ii)$] $W(\Pi) = W_V + W_L + \sum_i (n-n_i) \cdot S(i) + \sum_{(i,j)} m_{i,j}^L(n_i,n_j)$.
\end{itemize} 
\end{restatable}
\begin{proof}
Due to Corollary~\ref{cor:VertexSelection} we have
$$
P(\Pi_V)= P_V + \sum_i n_i \cdot L(i) + \sum_i (n-n_i) \cdot S(i).
$$  
Now, according to property~$(c)$ of Lemma~\ref{lem:LargeEdge}, we know that for each $1 \leq i < j \leq k$, $\Pi$ schedules early $n- n_i$ copies of~$f^{i,j}_1$, and $n-n_j$ copies of~$f^{i,j}_0$. According to property $(b)$ of \Cref{lem:LargeEdge}, $\Pi$ schedules exactly one jobs from $\{y^{i,j}_\ell,y^{i,j}_\ell\}$ for each $\ell \in \{1, \ldots, m\}$. By construction these jobs have a total processing time of
$$
\sum_{(i,j)}  m \cdot Y_{i,j} + \sum_i (n-n_i) \cdot L(i),
$$
so in total we have
\begin{align*}
P(\Pi) \,\,&=\,\, P_V + \sum_{(i,j)}  m \cdot Y_{i,j} + \sum_i n_i \cdot L(i) + \sum_i (n-n_i) \cdot L(i)   + \sum_i (n-n_i) \cdot S(i)  \\
&= P_V + P_L + \sum_i (n-n_i) \cdot S(i).
\end{align*}
For the total weight of early jobs in $\Pi$ we have
$$
W(\Pi_V)=W_V + \sum_i n_i \cdot L(i) + \sum_i (n-n_i) \cdot S(i)
$$ 
according to Corollary~\ref{cor:VertexSelection}. The total weight of all large edge jobs is  
$$
\sum_{(i,j)}  \sum_\ell Y_{i,j}\,/\, N^\ell  + \sum_i (n-n_i) \cdot L(i) + \sum_{(i,j)} m_{i,j}^L(n_i,n_j),
$$
including all filler jobs, as each edge in $E_{i,j}$ which is lexicographically greater or equal to some~$(n_i,n_j)$ contributes an additional unit to $W(\Pi)$ according to property~$(c)$ of Lemma~\ref{lem:LargeEdge}. Thus, all together we have
\begin{align*}
W(\Pi) \,\,&=\,\, W_V + \sum_{(i,j)}  \sum_\ell Y_{i,j}\,/\, N^\ell + \sum_i n_i \cdot L(i) + \sum_i (n-n_i) \cdot L(i) \\  
&+\,\, \sum_i (n-n_i) \cdot S(i) + \sum_{(i,j)} m_{i,j}^L(n_i,n_j) \\
&= W_V + W_L + \sum_i (n-n_i) \cdot S(i) + \sum_{(i,j)} m_{i,j}^L(n_i,n_j).
\end{align*}
and the corollary follows.
\end{proof}

\begin{example}
\label{ex:LargeEdge}
Recall the schedule of Example~\ref{ex:VertexSelection}. After optimally scheduling all jobs from the $(2,3)$ large edge gadget (including the filler jobs), the total processing time of all early jobs is:
$$
444| \underbrace{4000 44|}_{\substack{(2,3)\\ \text{large}}} | \underbrace{ 0000 13|}_{\substack{(1,3)\\ \text{large}}} | \underbrace{0000 12|}_{\substack{(1,2)\\ \text{large}}} | \underbrace{0000 21|}_{\substack{(2,3)\\ \text{small}}} | \underbrace{0000 31|}_{\substack{(1,3)\\ \text{small}}} | \underbrace{0000 32|}_{\substack{(1,2)\\ \text{small}}} | 0
$$
The total weight of all early jobs is 
$$
888| \underbrace{1111 44|}_{\substack{(2,3)\\ \text{large}}} | \underbrace{ 0000 13|}_{\substack{(1,3)\\ \text{large}}} | \underbrace{0000 12|}_{\substack{(1,2)\\ \text{large}}} | \underbrace{0000 21|}_{\substack{(2,3)\\ \text{small}}} | \underbrace{0000 31|}_{\substack{(1,3)\\ \text{small}}} | \underbrace{0000 32|}_{\substack{(1,2)\\ \text{small}}} | 2
$$    
as $m_{2,3}^L(2,3)=2$ in the example.
\end{example}

\subsection{Small edge gadget}
\label{sec:small-edge}

We next describe the small edge gadget. Analogous to the large edge gadget, the role of the small edge gadget is to  count all edges that are lexicographically smaller or equal to pairs of selected vertices. It is constructed similarly to the large edge gadget, except that we focus on the small blocks of the integers. We start by defining~$P^S_{i,j}$, which is analogous to the value~$P^L_{i,j}$ used in the large edge gadget:
$$
P^S_{i,j} \,\,\, =\,\,\,   \sum_{(i_0,j_0) > (i,j)} \left(m \cdot Z_{i_0,j_0} + n \cdot N^{G(i_0,j_0)+1} +  n \cdot N^{G(i_0,j_0)} \right) 
$$

For each $1 \leq i < j \leq j$, we construct the \emph{$(i,j)$ small edge gadget} as follows: Let $\ell \in \{1,\ldots,m\}$, and suppose that $e^{i,j}_\ell = \{v^i_{\ell_i},v^j_{\ell_j}\}\in E_{i,j}$ is the $\ell$'th edge in~$E_{i,j}$. We construct two jobs $z^{i,j}_\ell$ and $\neg z^{i,j}_\ell$ associated with $e^{i,j}_\ell$ that have the following characteristics:
\begin{itemize}
\item $p(z^{i,j}_\ell)=Z_{i,j}$ and $w(z^{i,j}_\ell)= Z_{i,j}\,/\,N^\ell+ 1$.
\item $p(\neg z^{i,j}_\ell)=Z_{i,j}$ and $w(\neg z^{i,j}_\ell)= Z_{i,j}\,/\,N^\ell$.
\item $d(z^{i,j}_\ell) = P_V + P_L + P^S_{i,j} +  \ell \cdot Z_{i,j} + (n-\ell_i) \cdot N^{G(i,j)+1} + (n-\ell_j) \cdot N^{G(i,j)} + N^{G(i,j)-1}$.
\item $d(\neg z^{i,j}_\ell) = P_V + P_L + P^S_{i,j} +  \ell \cdot Z_{i,j} + n \cdot N^{G(i,j)+1} + n \cdot N^{G(i,j)} + N^{G(i,j)-1}$.
\end{itemize}

We will also add \emph{filler jobs} as done in the large edge gadgets. We construct $n$ copies of the job pair $\{g^{i,j}_0,g^{i,j}_1\}$ which have the following characteristics:
\begin{itemize}
\item $p(g^{i,j}_0)=w(g^{i,j}_0)=N^{G(i,j)}$. 
\item $p(g^{i,j}_1)=w(g^{i,j}_1)=N^{G(i,j)+1}$.
\item $d(g^{i,j}_0) = d(g^{i,j}_1) = P_V+P_L+ P^S_{i,j}+ \ell \cdot Z_{i,j} + n \cdot N^{G(i,j)+1} + n \cdot N^{G(i,j)}$.
\end{itemize}
Thus, altogether, the $(i,j)$ small edge gadget consists of all job pairs $\{z^{i,j}_\ell,\neg z^{i,j}_\ell\}$ for $\ell \in \{1,\ldots,m\}$, and all $n$ copies of the job pair $\{g^{i,j}_0,g^{i,j}_1\}$. Note that all these jobs have three different processing times in total.

Lemma~\ref{lem:SmallEdge} below is analogous to Lemma~\ref{lem:LargeEdge}, except that the structure that it conveys allows us to count all edges that are lexicographically smaller or equal (as opposed to larger or equal) to pairs of selected vertices. 
We have the following: 

\begin{lemma}
\label{lem:SmallEdge}%
Let~$\Pi_V=(n_1,\ldots,n_k)$ be a schedule for the vertex selection jobs for some $n_1,\ldots,n_k \in \{1,\ldots,n\}$, and let $\Pi_L$ be an optimal extension of $\Pi_V$ to the set of large edge jobs. Let $\Pi$ be an optimal extension of $\Pi_L$ to the set of small edge jobs. Then the following properties hold for each~$1 \leq i < j \leq k$: 
\begin{itemize}
\item[(a)] The total processing time $P$ of all vertex selection jobs, all large edge jobs, and all $(i_0,j_0)$ small jobs for $(i_0,j_0) > (i,j)$, which are early in $\Pi$ satisfies
$$
P \,\, \geq \,\, P_V + P_L+ P^S_{i,j} + (n-n_i) \cdot N^{G(i,j)+1} + (n-n_j) \cdot N^{G(i,j)} 
$$
and 
$$
P \,\, \leq \,\, P_V + P_L + P^S_{i,j} + (n-n_i) \cdot N^{G(i,j)+1} + (n-n_j) \cdot N^{G(i,j)} + N^{G(i,j)-1}.
$$
\item[(b)] For each $\ell \in \{1,\ldots,m\}$ we have that either job $z^{i,j}_\ell$ or job $\neg z^{i,j}_\ell$ is early in $\Pi$, but not both. Job~$z^{i,j}_\ell$ is early if and only if $(n_i,n_j) \geq (\ell_i,\ell_j)$, where $e^{i,j}_\ell=\{v^i_{\ell_i},v^j_{\ell_j}\}$ is the $\ell$'th edge in~$E_{i,j}$.
\item[(c)] Precisely $n_i$ copies of job $g^{i,j}_1$ and $n_j$ copies of job $g^{i,j}_0$ are scheduled early in~$\Pi$.
\end{itemize}
\end{lemma}

\begin{proof}
The prove is via backward induction on $(i,j)$, starting with $(i,j)=(k-1,k)$. Consider first property $(a)$: Note that for $(i,j)=(k-1,k)$ we have $P=P(\Pi_L)$. As $P^S_{k-1,k}=0$, by Corollary~\ref{cor:LargeEdge} we have:
\begin{align*}
P = P(\Pi_L) \,&= \,  P_V + P_L + P^S_{k-1,k} + \sum_i  (n-v_i) \cdot S(i)\\
&\geq\, P_V + P_L + P^S_{k-1,k} + (n - n_{k-1}) \cdot N^{G(k-1,k)+1} + (n - n_k) \cdot N^{G(k-1,k)}. 
\end{align*}
On the other hand, $P$ is maximized if $\Pi_V$ selects vertex $1$ from each color class $V_i \neq V_{k-1},V_k$. As $N^{G(k-1,k)-1} > (n - 1) \cdot \sum^{k-2}_{i=1} S(i) + (n - n_{k-1}) \cdot \sum_{j=1}^{k-2} N^{G(j, k-1)} + (n- n_k) \cdot \sum_{j=1}^{k-2} N^{G(j, k)}$, we still have in this case that 
\begin{align*}
P \,\,&\leq\,\, P_V + P_L + (n-n_{k-1}) \cdot S(k-1) + (n-n_k) \cdot S(k) + (n-1) \cdot \sum^{k-2}_{i=1} S(i) \\
&\leq\,\, P_V + P_L + P^S_{k-1,k} + (n-n_{k-1}) \cdot N^{G(k-1, k) + 1} + (n-n_k) \cdot N^{G(k-1, k)} + N^{G(k-1,k)-1},
\end{align*}
and so property $(a)$ indeed holds.

Let us next prove property $(b)$ by induction on $\ell$. Let $\ell =1$. As is done in the proof of Lemma~\ref{lem:LargeEdge}, one can show that $\Pi$ schedules early exactly one job from the job pair~$\{z^{k-1,k}_1,\neg z^{k-1,k}_1\}$, and $z^{k-1,k}_1$ is scheduled early iff $P + p(z^{k-1,k}_1) \leq d(z^{k-1,k}_1)$. If $(n_{k-1},n_k) \geq (\ell_{k-1},\ell_k)$, then clearly~$(n-n_{k-1},n-n_k) \leq (n-\ell_{k-1},n-\ell_k)$, and so by property~$(a)$ we have
\begin{align*}
P + p(z^{k-1,k}_1) \,\, &\leq \,\, P_V \!+\! P_L  \!+\! (n-n_{k-1}) \!\cdot\! N^{G(k-1,k)+1} \!+\! (n-n_k) \!\cdot\! N^{G(k-1,k)} \!+\! N^{G(k,k-1)-1} \!+\! Z_{k-1,k} \\
&\leq \,\, P_V \!+\! P_L \!+\! (n-\ell_{k-1}) \!\cdot\! N^{G(k-1,k)+1} \!+\! (n-\ell_k) \!\cdot\! N^{G(k-1,k)} \!+\! N^{G(k-1,k)-1} \!+\! Z_{k-1,k}\\  
&=\,\, d(z^{k-1,k}_1).
\end{align*}
Otherwise, if $(n_{k-1},n_k) < (\ell_{k-1},\ell_k)$ then $(n-n_{k-1},n-n_k) > (n-\ell_{k-1},n-\ell_k)$, and so again by property~$(a)$ we have
\begin{align*}
P + p(z^{k-1,k}_1) \,\, &\geq \,\, P_V+P_L + (n-n_{k-1}) \cdot N^{G(k-1,k)+1} + (n-n_k) \cdot N^{G(k-1,k)} +Z_{k-1,k} \\
&>\,\,  P_V+P_L  + (n-n_{k-1}) \cdot N^{G(k-1,k)+1} + (n-n_k-1) \cdot N^{G(k-1,k)}\\ 
&+\,\, N^{G(k-1,k)-1} +Z_{k-1,k}\\
&\geq\,\,  P_V+P_L  + (n-\ell_{k-1}) \cdot N^{G(k-1,k)+1} + (n-\ell_k) \cdot N^{G(k-1,k)}\\ 
&+\,\, N^{G(k-1,k)-1}+Z_{k-1,k} \,\,=\,\, d(z^{k-1,k}_1). 
\end{align*}
Thus, property~$(b)$ holds for $\ell=1$. The inductive step for $\ell >1$ follows by the exact same arguments while observing that both $d(z^{i,j}_\ell)$ and $d(\neg z^{i,j}_\ell)$ contain the term $\ell \cdot Z_{i,j}$. 

Finally, for property~$(c)$, as argued in the proof of Lemma~\ref{lem:LargeEdge}, we know that all filler jobs are scheduled in~$\Pi$ at time~$P^*=P+m \cdot Z_{k-1,k}$, where first all filler jobs of type $g^{k-1,k}_1$ are scheduled, followed by all filler jobs of type $g^{k-1,k}_0$. According to property~$(a)$, we can schedule~$n_{k-1}$ copies of job~$g^{k-1,k}_1$ since
\begin{align*}
P^*+n_{k-1} \cdot p(g^{k-1,k}_1) \,\, &\leq \,\, P_V+P_L + m \cdot Z_{k-1,k} \\
&+\,\, n \cdot N^{G(k-1,k)+1} + (n-n_k) \cdot N^{G(k-1,k)} + N^{G(k-1,k)-1} \\
&<\,\, P_V+P_L + m \cdot Z_{k-1,k} \\
&+\,\, n \cdot N^{G(k-1,k)+1} + n \cdot N^{G(k-1,k)} + N^{G(k-1,k)-1} \,\,=\,\, d(g^{k-1,k}_1).    
\end{align*}
Scheduling more than $n_{k-1}$ copies is not possible since $P^* + (n_{k-1}+1) \cdot p(g^{k-1,k}_1)$ is at least 
$$
P_V+P_L +m \cdot Z_{k-1,k}+ (n+1) \cdot N^{G(k-1,k)+1} + n_k \cdot N^{G(k-1,k)} > d(g^{k-1,k}_1). 
$$
Thus, $\Pi$ schedules precisely $n_{k-1}$ copies of~$g^{k-1,k}_1$. A similar argument shows that~$\Pi$ also schedules precisely $n_k$ copies of~$g^{k-1,k}_0$.

This completes the proof of the base case $(i,j)=(k-1,k)$. The inductive step for property~$(a)$ follows by observing that by induction the total processing time of all early $(i_0,j_0)$ small edge jobs in~$\Pi$ with $(i_0,j_0) > (i,j)$ is exactly 
$$
P^S_{i,j} \,\,\,- \sum_{(i_0, j_0) > (i,j)} \bigl(n_{i_0} \cdot N^{G(i_0,j_0) + 1} + n_{j_0} \cdot N^{G(i_0,j_0)}\bigr).
$$
Properties $(b)$ and~$(c)$ then follow using the same arguments as above.
\end{proof}

For $1 \leq i < j \leq k$ and~$n_i,n_j \in \{1,\ldots,n\}$, define $m^S_{i,j}(n_i,n_j)$ to be the total number of edges in $E_{i,j}$ that are lexicographically smaller or equal to $(n_i,n_j)$. That is, the total number of edges $e^{i,j}_\ell=(v^i_{\ell_i},v^j_{\ell_j})\in E_{i,j}$ with $(n_i,n_j) \geq (\ell_i,\ell_j)$. Let~$W_S$ denote the following value: 
$$
W_S = \sum_{(i,j)}  \left(\sum_\ell Z_{i,j} \,/\, N^\ell + n \cdot N^{G(i,j)+1} +  n \cdot N^{G(i,j)} \right).
$$
We have the following corollary of Lemma~\ref{lem:SmallEdge}.

\begin{restatable}[]{corollary}{SmallEdge}
\label{cor:SmallEdge}
Let~$\Pi_V=(n_1,\ldots,n_k)$ be a schedule for the vertex selection jobs for some $n_1,\ldots,n_k \in \{1,\ldots,n\}$, let $\Pi_L$ be an optimal extension of $\Pi_V$ to the set of large edge jobs, and let $\Pi$ be an optimal extension of $\Pi_L$ to the set of small edge jobs. Then
$$
W(\Pi)  = W_V + W_L + W_S +\sum_{(i,j)} m_{i,j}^L(n_i,n_j) + \sum_{(i,j)} m_{i,j}^S(n_i,n_j).
$$
\end{restatable}

\begin{proof}
By Corollary~\ref{cor:LargeEdge} we have 
$$
W(\Pi_L) = W_V + W_L + \sum_i (n-n_i) \cdot S(i) + \sum_{(i,j)} m_{i,j}^L(n_i,n_j).
$$
due to Corollary~\ref{cor:LargeEdge}. The total weight of all early small jobs is 
$$
\sum_{(i,j)}  \sum_\ell Z_{i,j} \,/\, N^\ell  + \sum_i n_i \cdot S(i) + \sum_{(i,j)} m_{i,j}^S(n_i,n_j),
$$
since each edge in $E_{i,j}$ which is lexicographically smaller or equal to $(v_i,v_j)$ contributes an additional unit to $W(\Pi)$ according to property~$(c)$ of Lemma~\ref{lem:SmallEdge}. Thus, all together we have
\begin{align*}
W(\Pi) \,\,&= W_V + W_L + \sum_{(i,j)}  \sum_\ell Z_{i,j} \,/\, N^\ell + \sum_i n \cdot S(i) + \sum_{(i,j)} m_{i,j}^L(n_i,n_j) + \sum_{(i,j)} m_{i,j}^S(n_i,n_j) \\
&=\,\, W_V+ W_L + W_S +\sum_{(i,j)} m_{i,j}^L(n_i,n_j) + \sum_{(i,j)} m_{i,j}^S(n_i,n_j),
\end{align*}
and so the corollary follows.
\end{proof}

\begin{example}
Consider the schedule of Example~\ref{ex:LargeEdge}. After scheduling all remaining large edge jobs, and all jobs from the $(2,3)$ small edge gadget (including the filler jobs), the total processing time of all early jobs is:
$$
444| \underbrace{4000 44|}_{\substack{(2,3)\\ \text{large}}} | \underbrace{ 4000 44|}_{\substack{(1,3)\\ \text{large}}} | \underbrace{4000 44|}_{\substack{(1,2)\\ \text{large}}} | \underbrace{4000 44|}_{\substack{(2,3)\\ \text{small}}} | \underbrace{0000 31|}_{\substack{(1,3)\\ \text{small}}} | \underbrace{0000 32|}_{\substack{(1,2)\\ \text{small}}} | 0
$$
The total weight of all early jobs is 
$$
888| \underbrace{1111 44|}_{\substack{(2,3)\\ \text{large}}} | \underbrace{ 1111 44|}_{\substack{(1,3)\\ \text{large}}} | \underbrace{1111 44|}_{\substack{(1,2)\\ \text{large}}} | \underbrace{1111 4|}_{\substack{(2,3)\\ \text{small}}} | \underbrace{0000 31|}_{\substack{(1,3)\\ \text{small}}} | \underbrace{0000 32|}_{\substack{(1,2)\\ \text{small}}} | 12
$$    
as $\sum_{(i,j)} m_{i,j}^L(i,j) + m_{2,3}^S(2,3)=12$ in the example.
\end{example}

\subsection{Correctness}

We have completed the description of all jobs in our \prob instance. Table~\ref{tab:jobs} provides a compact list of the characteristics of all these jobs. Lemma~\ref{lem:Final} below, along with Theorem~\ref{thm:mcc}, completes our proof of Theorem~\ref{thm:main} for parameter~$p_{\#}$. 

\begin{table}[h!]
\captionsetup{singlelinecheck=off}
\begin{center}
\begin{tabular}{c | c | c| c}
Job & Processing Time & Weight & Due Date\\
\hline
$x^*_i $ & $X_i + L(i)$ & $(n+1) \cdot X_i + L(i)$ & $P^V_{i-1} + N^{(m+2) \cdot 2\binom{k}{2}}$\\
$x_i $ & $X_i + L(i)$ & $X_i + L(i)$ 
&$d(x^*_i)$\\
$\neg x_i $ & $X_i  + S(i)$ & $X_i + S(i)$ & $d(x^*_i)$\\
\hline
\hline
$y^{i,j}_\ell$ & $Y_{i,j}$ & $Y_{i,j} \,/\, N^\ell + 1$ & $P^L_{i,j}(\ell) + \ell_i \cdot N^{F(i,j)+1} + \ell_j \cdot N^{F(i,j)}$ \\
$\neg y^{i,j}_\ell$ & $Y_{i,j}$ & $Y_{i,j} \,/\, N^\ell$ & $P^L_{i,j}(\ell) + n \cdot N^{F(i,j)+1} + n \cdot N^{F(i,j)}$\\
\hline
$f^{i,j}_1$ & $N^{F(i,j)+1}$ & $N^{F(i,j)+1}$ & $P^L_{i,j}(m) + n \cdot N^{F(i,j)+1}+ n \cdot N^{F(i,j)}$\\
$f^{i,j}_0$ & $N^{F(i,j)}$ & $N^{F(i,j)}$ & $d(f^{i,j}_1)$\\
\hline
\hline
$z^{i,j}_\ell$ & $Z_{i,j}$ & $Z_{i,j} \,/\, N^\ell + 1$ & $P^S_{i,j}(\ell) + (n\!-\!\ell_i) \cdot N^{G(i,j)+1} + (n\!-\!\ell_j) \cdot N^{G(i,j)}$ \\
$\neg z^{i,j}_\ell $ & $Z_{i,j}$ & $Z_{i,j} \,/\, N^\ell$ & $P^S_{i,j}(\ell) + n \cdot N^{G(i,j)+1} + n \cdot N^{G(i,j)}$\\
\hline
$g^{i,j}_1$ & $N^{G(i,j)+1}$ & $N^{G(i,j)+1}$ & $P^S_{i,j}(m) + n \cdot N^{G(i,j)+1}+ n \cdot N^{G(i,j)}$\\
$g^{i,j}_0$ & $N^{G(i,j)}$ & $N^{G(i,j)}$ & $d(g^{i,j}_1)$
\end{tabular}
\end{center}
\caption[sdsfa]{The weights, processing times, and due dates of all jobs in our construction. Here, $P^L_{i,j}(\ell)$ is shorthand notation for $P_V+P^L_{i,j} + \ell \cdot Y_{i,j} + N^{F(i,j) -1}$, and $P^S_{i,j}(\ell) = P_V+P_L+P^S_{i,j} + \ell \cdot Z_{i,j} + N^{G(i,j) -1}$.}
\label{tab:jobs}
\end{table}

\begin{lemma}
\label{lem:Final}%
There is a parameterized reduction from \mcc (restricted to nice $k$-partite graphs) parameterized by $k$ to \prob parameterized by~$p_{\#}$.
\end{lemma}

\begin{proof}
The reduction is as described throughout the section. It is in fact a reduction to the equivalent problem of \prob where the goal is to maximize the weight of early jobs. The reduction can be carried out in polynomial-time, and the total number of different processing-times $p_{\#}$ in the resulting \prob instance is $2k+6\binom{k}{2}$ (see Table~\ref{tab:jobs}). To complete the proof of the lemma, we argue that the graph $G=(V=V_1 \uplus \cdots \uplus V_k,E)$ of the input \mcc instance has a clique of size $k$ iff the constructed \prob instance has a schedule where the total weight of early jobs is at least $W_V+W_L+W_S+ (m+1)\cdot\binom{k}{2}$.

Suppose $G$ has a clique of size $k$ with $v^1_{n_1}\in V_i,\ldots,v^k_{n_k} \in V_k$. Then $\sum_{(i,j)} m^L_{i,j}(n_i,n_j) + m^S_{i,j}(n_i,n_j) = (m+1)\cdot \binom{k}{2}$. Thus, according to Corollary~\ref{cor:SmallEdge}, the optimal extension $\Pi$ of $\Pi_V=\Pi(n_1,\ldots,n_k)$ to the set of large and small edge jobs has total weight of early jobs $W(\Pi) = W_V+W_L+W_S+(m+1)\cdot \binom{k}{2}$. Conversely, suppose that there is a schedule $\Pi$ for the \prob instance with $W(\Pi) \geq W_V+W_L+W_S+(m+1)\cdot \binom{k}{2}$. Let $\Pi_V$ be the restriction of $\Pi$ to the vertex selection jobs. Then as $N^{(m+2) \cdot 2 \binom{k}{2} -1}$ is larger than the total weight of all large and small edge jobs, we have $W(\Pi_V) \geq W_V$ as otherwise we have $W(\Pi_V) < W_V - X_1 + N^{(m+2) \cdot 2 \binom{k}{2} -1}$, implying
$W(\Pi) < W_V$. Thus, $\Pi_V=\Pi(n_1,\ldots,n_k)$ for some $n_1,\ldots,n_k \in \{1,\ldots,n\}$ according to Lemma~\ref{lem:VertexSelectBackward}. We may assume without loss of generality that $\Pi$ is an optimal extension of $\Pi_V$ to set of large and small edge jobs. It follows then from Corollary~\ref{cor:SmallEdge} that 
$\sum_{(i,j)} m^L_{i,j}(n_i,n_j) + m^S_{i,j}(n_i,n_j) = (m+1)\cdot \binom{k}{2}$, which means that there are $\binom{k}{2}$ edges in $G$ between vertices in $\{v^1_{n_1},\ldots,v^k_{n_k}\}$. Thus, $v^1_{n_1},\ldots,v^k_{n_k}$ is a clique of size $k$ in $G$.
\end{proof}


\section{Parameter \boldmath{$w_{\#}$}}
\label{sec:WSharp}%

In the following section we adapt the reduction from Section~\ref{sec:PSharp} to show that \prob is W[1]-hard parameterized by $w_{\#}$, completing the proof of \Cref{thm:main}.
On a high level, we adapt the reduction of the previous section by swapping processing times and weights of the jobs, except for the last digit (\emph{i.e.} the counting block). This then also causes adaptions on the due dates of the jobs.

\subsection{Vertex selection gadget}

This gadget is identical to the one from \Cref{sec:vertex-selection}.
Thus, \Cref{lem:VertexSelectForward,lem:VertexSelectBackward,cor:VertexSelection} directly transfer.

\subsection{Large edge gadget}
\label{sec:LargeEdgeModified}%

Let $1 \leq i < j \leq k$. The $(i,j)$ large edge gadget is based on the same idea as the one from \Cref{sec:large-edge}.
Note in~\Cref{sec:large-edge}, we used the varying weights to ensure that one of $\{y^{i,j}_\ell, \neg y^{i,j}_\ell\}$ is scheduled early prior to scheduling any job from $\{y^{i,j}_{\ell-1}, \neg y^{i,j}_{\ell-1}\}$. As this results in an unbounded~$w_{\#}$, we instead use varying processing times to ensure this. More specifically, we essentially swap the weights and profits of the jobs, with the exception of the counting block. This results in a reversed scheduling order within the gadget, ensuring that one of $\{y^{i,j}_{\ell}, \neg y^{i,j}_{\ell}\}$ is scheduled early prior to scheduling any job from $\{y^{i,j}_{\ell -1}, \neg y^{i,j}_{\ell-1}\}$.

We begin by adapting~$P_{i,j}^L$ to be the following value:
$$
P^L_{i,j} \,\,\,=\,\,\, \sum_{(i_0,j_0) > (i,j)} \left(\sum_\ell Y_{i_0,j_0} \,/\, N^\ell + n \cdot N^{F(i_0,j_0)+1} +  n \cdot N^{F(i_0,j_0)} \right)
$$
Let $\ell \in \{1,\ldots,m\}$, and suppose that the $\ell$'th edge between~$V_i$ and~$V_j$ is the edge $e^{i,j}_\ell =\{v^i_{\ell_i},v^j_{\ell_j}\}$ for some $\ell_i,\ell_j \in \{1,\ldots,n\}$. The two jobs~$y^{i,j}_\ell$ and~$\neg y^{i,j}_\ell$ corresponding to~$e^{i,j}_\ell$ are constructed with the following characteristics:
\begin{itemize}
\item $p(y^{i,j}_\ell)=Y_{i,j} / N^\ell$ and $w(y^{i,j}_\ell) = Y_{i,j} + 1$.
\item $p(\neg y^{i,j}_\ell)=Y_{i,j} / N^\ell$ and $w(\neg y^{i,j}_\ell) = Y_{i,j}$.
\item $d(y^{i,j}_\ell)= P_V + P^L_{i,j} +  \sum_{\ell_0 \geq \ell} Y_{i,j}/N^{\ell_0} + \ell_i \cdot N^{F(i,j)+1} + \ell_j \cdot N^{F(i,j)} + N^{F(i,j)-1}$.
\item $d(\neg y^{i,j}_\ell)= P_V + P^L_{i,j} +  \sum_{\ell_0 \geq \ell} Y_{i,j}/N^{\ell_0} + n \cdot N^{F(i,j)+1} + n \cdot N^{F(i,j)} + N^{F(i,j)-1}$.
\end{itemize}
Note that $d(y^{i,j}_\ell) < d(y^{i,j}_{\ell -1})$ for each $1 < \ell \leq m$, and consequently, job~$y^{i,j}_\ell$ or $\neg y^{i,j }_\ell$ will be scheduled prior to $y^{i,j}_{\ell -1}$ or $\neg y^{i,j }_{\ell-1}$.

The filler jobs are constructed similarly to the previous section, except that their due dates need to be adjusted according to the modified processing times of jobs $y^{i,j}_\ell$ and $\neg y^{i,j}_\ell$:
\begin{itemize}
\item $p(f^{i,j}_0)=w(f^{i,j}_0)=N^{F(i,j)}$.
\item $p(f^{i,j}_1)=w(f^{i,j}_1)=N^{F(i,j)+1}$.
\item $d(f^{i,j}_0) = d(f^{i,j}_1) = P_V + P^L_{i,j} +  \sum_{\ell_0 \geq \ell} Y_{i,j}/N^{\ell_0} + n \cdot N^{F(i,j)+1} + n \cdot N^{F(i,j)} + N^{F(i,j)-1}$.
\end{itemize}

We also need to adapt $W_L$, as the weights of the large edge jobs have been modified:
$$
W_L = \sum_{(i,j)} \left( m \cdot Y_{i,j} + n\cdot N^{F(i,j) + 1} + n\cdot N^{F(i,j)}\right)
$$
Following these modifications, we now show \Cref{lem:LargeEdge} where only the proof of property $(b)$ differs from \Cref{sec:WSharp}:

\lemLargeEdge*

\begin{proof}
We prove that~$\Pi$ satisfies the properties of lemma by backward induction on $(i,j)$, starting with the base case of $(i,j)=(k-1,k)$. 

Property $(a)$ is shown as in \Cref{sec:PSharp}.

We next prove property $(b)$ for $(i,j)=(k-1,k)$.
Note that as both $w(y^{k-1,k}_\ell)$ and~$w(\neg y^{k-1,k}_\ell) \geq N^{F(k-1,k)+m+1}$ are larger than the total weight of all other large edge jobs, $m$ jobs from $\{y^{k-1, k}_\ell, \neg y^{k-1, k}_\ell : \ell \in [m]\}$ must be early in~$\Pi$, as otherwise $\Pi$ is not an optimal extension of $\Pi_V$.
As $\Pi$ is an EDD schedule, no job $y^{k-1, k}_\ell$ or $\neg y^{k-1, k}_\ell$ is scheduled before time $P$.
Note that for each $\ell \in [m]$, at most one of $y^{k-1, k}_\ell$ and $\neg y^{k-1, k}_\ell $ can be early as
\begin{align*}
P+p(y^{k-1,k}_\ell)& + p(\neg y^{k-1,k}_\ell) \,\, > \,\, P_V + p(y^{k-1,k}_\ell) + p(\neg y^{k-1,k}_\ell) 
=\,\, P_V +2 \cdot Y_{k-1,k}/ N^\ell\\ 
& > \,\, P_V + \sum_{\ell_0  > \ell} Y_{k-1,k}/ N^{\ell_0} + n \cdot N^{F(k-1,k)+1} + n \cdot N^{F(k-1,k)} + N^{F(k-1,k)-1}\\ 
&=\,\, d (\neg y^{k-1,k}_1) \,\, \ge \,\, d(y^{k-1,k}_1)\,.
\end{align*}
Consequently, for each $\ell \in [m]$, exactly one of $y^{k-1,k}_\ell$ and $\neg y^{k-1,k}_\ell$ is early.
As $\Pi$ is an EDD schedule, the early job of~$y^{k-1,k}_\ell$ and $\neg y^{k-1,k}$ is scheduled before all jobs~$y^{k-1,k}_{\ell_0}$ or $y^{k-1, k}_{\ell_0}$ for $\ell >  \ell_0$.
Scheduling exactly one of $y^{k-1, k}_\ell$ or $\neg y^{k-1, k}_\ell$ is always possible as
\begin{align*}
P + \sum_{\ell_0 \ge \ell} p(\neg y^{k-1,k}_{\ell_0}) \,\,&\leq \,\,  P_V + \sum_{\ell_0 \ge \ell} Y_{k-1, k} / N^{\ell_0} + n_{k-1} \cdot N^{F(k-1,k)+1} + n_k \cdot N^{F(k-1,k)} + N^{F(k-1,k)-1}\\
&\leq\,\, P_V  + \sum_{\ell_0 \ge \ell} Y_{k-1, k} / N^{\ell_0} + n \cdot N^{F(k-1,k)+1} + n \cdot N^{F(k-1,k)} + N^{F(k-1,k)-1} \\
\,\,& = \,\, d(\neg y^{k-1,k}_1)\,.
\end{align*}

As $w(y^{k-1,k}_\ell) > w(\neg y^{k-1,k}_\ell)$ and $p(y^{k-1, k}_\ell) = p(\neg y^{k-1, k}_\ell)$, an optimal extension of~$\Pi_V$ would schedule job~$y^{k-1,k}_\ell$ early if possible. If $(n_{k-1},n_k) \leq (\ell_{k-1},\ell_{k})$ then
\begin{align*}
P + \sum_{\ell_0 \ge \ell} p(y^{k-1,k}_{\ell_0}) \,\, &\leq \,\, P_V + n_{k-1} \cdot N^{F(k-1,k)+1} + n_k \cdot N^{F(k-1,k)} + N^{F(k-1,k)-1}+ \sum_{\ell_0 \ge \ell} Y_{k-1, k} / N^{\ell_0} \\
&\leq \,\, P_V + \ell_{k-1} \cdot N^{F(k-1,k)+1} + \ell_k \cdot N^{F(k-1,k)} + N^{F(k-1,k)-1} + \sum_{\ell_0 \ge \ell} Y_{k-1, k} / N^{\ell_0} \\
& \,\, = d(y^{k-1,k}_1),
\end{align*}
and so~$y^{k-1,k}_\ell$ is indeed early in~$\Pi$. If $(n_{k-1},n_k) > (\ell_{k-1},\ell_k)$, then
\begin{align*}
P + \sum_{\ell_0 \ge \ell} p(y^{k-1,k}_{\ell_0}) \,\, &\geq \,\, P_V + n_{k-1} \cdot N^{F(k-1,k)+1} + n_k \cdot N^{F(k-1,k)} + \sum_{\ell_0 \ge \ell} Y_{k-1, k} / N^{\ell_0} \\
&>\,\, P_V + n_{k-1} \cdot N^{F(k-1,k)+1} + (n_k-1) \cdot N^{F(k-1,k)} + N^{F(k-1,k)-1} +Y_{k-1,k}\\
&\geq\,\, P_V + \ell_{k-1} \cdot N^{F(k-1,k)+1} + \ell_k \cdot N^{F(k-1,k)} + N^{F(k-1,k)-1}+ \sum_{\ell_0 \ge \ell} Y_{k-1, k} / N^{\ell_0}\\
\,\,&=\,\, d(y^{k-1,k}_1), 
\end{align*}
and so~$y^{k-1,k}_\ell$ is not early in~$\Pi$. Thus, property~$(b)$ holds.

Property $(c)$ follows as in \Cref{sec:PSharp}.

We have thus shown that the lemma holds for the base case of~$(i,j)=(k-1,k)$. The inductive step for property~$(a)$ follows by observing that, by induction, the total processing time of all early $(i_0,j_0)$ large edge jobs in~$\Pi$ with $(i_0,j_0) > (i,j)$ is exactly 
$$
P^L_{i,j} \,\,\,- \sum_{(i_0, j_0) > (i,j)} \bigl(n_{i_0} \cdot N^{F(i_0,j_0) + 1} + n_{j_0} \cdot N^{F(i_0,j_0)}\bigr).
$$
Properties $(b)$ and~$(c)$ then follow using the same arguments as above.
\end{proof}
Having shown \Cref{lem:LargeEdge}, \Cref{cor:LargeEdge} now follows with the same proof as in \Cref{sec:PSharp}.
\LargeEdge*

\subsection{Small edge gadget}

Let $1 \leq i < j \leq k$. The $(i,j)$ small edge gadget is modified similarly to the $(i,j)$ large edge gadget. We adapt~$P_{i,j}^S$ to be the following value:
$$
P^S_{i,j} \,\,\,=\,\,\, \sum_{(i_0,j_0) > (i,j)} \left(\sum_\ell Z_{i_0,j_0} \,/\, N^\ell + n \cdot N^{G(i_0,j_0)+1} +  n \cdot N^{G(i_0,j_0)} \right)
$$
The two jobs~$z^{i,j}_\ell$ and~$\neg z^{i,j}_\ell$ corresponding to edge~$e^{i,j}_\ell =\{v^i_{\ell_i},v^j_{\ell_j}\}$ are constructed with the following characteristics:
\begin{itemize}
\item $p(z^{i,j}_\ell)=Z_{i,j} / N^\ell$ and $w(z^{i,j}_\ell) = Z_{i,j} + 1$.
\item $p(\neg z^{i,j}_\ell)=Z_{i,j} / N^\ell$ and $w(\neg z^{i,j}_\ell) = Z_{i,j}$.
\item $d(z^{i,j}_\ell)= P_V + P_L+P^S_{i,j} +  \sum_{\ell_0 \geq \ell} Z_{i,j}/N^{\ell_0} + (n-\ell_i) \cdot N^{G(i,j)+1} + (n-\ell_j) \cdot N^{G(i,j)} + N^{G(i,j)-1}$.
\item $d(\neg z^{i,j}_\ell)= P_V + P_L+P^S_{i,j} +  \sum_{\ell_0 \geq \ell} Z_{i,j}/N^{\ell_0} + n \cdot N^{G(i,j)+1} + n \cdot N^{G(i,j)} + N^{G(i,j)-1}$.
\end{itemize}
The filler jobs $g^{i,j}_0$ and $g^{i,j}_1$ are constructed as follows:
\begin{itemize}
\item $p(g^{i,j}_0)=w(g^{i,j}_0)=N^{G(i,j)}$.
\item $p(g^{i,j}_1)=w(g^{i,j}_1)=N^{G(i,j)+1}$.
\item $d(g^{i,j}_0) = d(g^{i,j}_1) = P_V + P_L+P^S_{i,j} +  \sum_{\ell_0 \geq \ell} Z_{i,j}/N^{\ell_0} + n \cdot N^{F(i,j)+1} + n \cdot N^{F(i,j)} + N^{F(i,j)-1}$.
\end{itemize}

We adapt $W_S$ to be following value:
$$
W_S = \sum_{(i,j)} \left( m \cdot Z_{i,j} + n\cdot N^{G(i,j) + 1} + n\cdot N^{G(i,j)}\right)
$$
\Cref{cor:SmallEdge} now applies after all modifications above, with the same modifications in the proof as for the large edge gadget. \SmallEdge*

\subsection{Correctness}

This completes the description of the modified construction. A compact description of the weights, processing times, and due dates of all jobs can be found in \Cref{tab:wsharp-jobs}. Using the analogous version of \Cref{cor:SmallEdge}, an analog of \Cref{lem:Final} given in \Cref{lem:FinalModified} below follows with the same proof.
This, together with \Cref{thm:mcc}, finishes the proof of \Cref{thm:main} for parameter~$w_\#$. 

\begin{lemma}
\label{lem:FinalModified}%
There is a parameterized reduction from \mcc (restricted to nice $k$-partite graphs) parameterized by $k$ to \prob parameterized by~$w_{\#}$.
\end{lemma}

\begin{table}
\captionsetup{singlelinecheck=off}
\begin{center}
\begin{tabular}{c | c | c| c}
Job & Processing Time & Weight & Due Date\\
\hline
$x^*_i $ & $X_i + L(i)$ & $(n+1) \cdot X_i + L(i)$ & $P^V_{i-1} + N^{(m+2) \cdot 2\binom{k}{2}}$\\
$x_i $ & $X_i + L(i)$ & $X_i + L(i)$ 
&$d(x^*_i)$\\
$\neg x_i $ & $X_i  + S(i)$ & $X_i + S(i)$ & $d(x^*_i)$\\
\hline
\hline
$y^{i,j}_\ell$ & $Y_{i,j} \,/\, N^\ell$ & $Y_{i,j} + 1$ & $P^L_{i,j}(\ell) + \ell_i \cdot N^{F(i,j)+1} + \ell_j \cdot N^{F(i,j)}$ \\
$\neg y^{i,j}_\ell$ & $Y_{i,j} \,/\, N^\ell$ & $Y_{i,j}$ & $P^L_{i,j}(\ell) + n \cdot N^{F(i,j)+1} + n \cdot N^{F(i,j)}$\\
\hline
$f^{i,j}_1$ & $N^{F(i,j)+1}$ & $N^{F(i,j)+1}$ & $P^L_{i,j}(m) + n \cdot N^{F(i,j)+1}+ n \cdot N^{F(i,j)}$\\
$f^{i,j}_0$ & $N^{F(i,j)}$ & $N^{F(i,j)}$ & $d(f^{i,j}_1)$\\
\hline
\hline
$z^{i,j}_\ell$ & $Z_{i,j} \,/\, N^\ell$ & $Z_{i,j} + 1$ & $P^S_{i,j}(\ell) + (n\!-\!\ell_i) \cdot N^{G(i,j)+1} + (n\!-\!\ell_j) \cdot N^{G(i,j)}$ \\
$\neg z^{i,j}_\ell $ & $Z_{i,j} \,/\, N^\ell$ & $Z_{i,j}$ & $P^S_{i,j}(\ell) + n \cdot N^{G(i,j)+1} + n \cdot N^{G(i,j)}$\\
\hline
$g^{i,j}_1$ & $N^{G(i,j)+1}$ & $N^{G(i,j)+1}$ & $P^S_{i,j}(m) + n \cdot N^{G(i,j)+1}+ n \cdot N^{G(i,j)}$\\
$g^{i,j}_0$ & $N^{G(i,j)}$ & $N^{G(i,j)}$ & $d(g^{i,j}_1)$
\end{tabular}
\end{center}
\caption[]{Weights, processing times, and due dates for our reduction showing W[1]-hardness with respect to~$w_{\#}$.
Note that constants~$P_{i,j}^L (\ell)$ and $P_{i,j}^S (\ell)$ are different from \Cref{sec:PSharp}.
}
\label{tab:wsharp-jobs}
\end{table}

\section{ETH Lower Bounds}
\label{sec:ETH}%

In the following we show that Theorem~\ref{thm:main} implies that the $O(n^{p_\#+1} \lg n)$ and $O(n^{w_{\#}+1} \lg n)$ time algorithms from~\cite{HermelinKPS21} are close to being optimal under the Exponential Time Hypothesis (ETH). Note that since \textsc{Multicolored Clique} cannot be solved in $n^{o(k)}$ time under ETH~\cite{ChenHKX06}, our reduction from \Cref{sec:PSharp,sec:WSharp} directly implies that there is no $n^{o(\sqrt{p_{\#}})}$ or $n^{o(\sqrt{w_{\#}})}$-time algorithm for \prob, as our reduction uses $O(k^2)$ many different processing times or weights. We will now sketch how to improve these lower bounds to $n^{o(p_{\#}/\lg p_\#)}$ and $n^{o(w_\# / \lg w_\#)}$.
In order to do so, we use the standard trick to reduce from \textsc{Partitioned Subgraph Isomorphism} instead of \textsc{Multicolored Clique} (see also \cite{KarthikMPS23} for a list of papers using this trick).
\begin{definition}
Given an $\ell$-partite graph $G = (V_1 \uplus \ldots \uplus V_\ell, E)$ and a graph~$H$ on the vertex set~$\{1,\ldots,\ell\}$, the \textsc{Partitioned Subgraph Isomorphism} problem asks to determine whether~$G$ contains a vertex subset~$\{v_1,  \ldots, v_\ell\}$ with $v_i\in V_i$ for any $i \in \{1,\ldots,\ell\}$ such that~$\{v_i, v_j\} \in E(G)$ whenever $\{i,j\} \in E(H)$.
\end{definition}
\noindent Our almost tight ETH-based lower bounds for \prob are based on the following lower bound for \textsc{Partitioned Subgraph Isomorphism} due to Marx~\cite{Marx10}:
\begin{theorem}[\cite{Marx10}]
\textsc{Partitioned Subgraph Isomorphism} cannot be solved in $f (k) \cdot n^{o(k/ \lg k)}$ time, where $f$ is an arbitrary
function and $k = |E(H)|$ is the number of edges of the smaller graph $H$, unless ETH is false.
\end{theorem}

As for \textsc{Multicolored Clique}, we may assume without loss of generality that the graph~$G$ in \textsc{Partitioned Subgraph Isomorphism} is nice. (This can be achieved by adding isolated vertices and edges between isolated vertices.) The key observation we use is that we can remove all large and small $(i,j)$ blocks from all integers considered in our construction, for any $1 \leq i < j \leq k$ with $\{i,j\} \notin E(H)$. This is because now we do not need to ensure that there is an edge between $v^i_{n_i} \in V_i$ and~$v^i_{n_j} \in  V_j$. Practically speaking, we remove all $(i,j)$ large and small edge jobs for each $(i,j)$ with $\{i,j\} \notin E(H)$, thereby reducing the number of different processing times (respectively weights) to $2\ell + 3k$ (respectively $2\ell+ 4k$). Apart from this, we also need to adapt the due dates of the remaining large and small edge jobs by subtracting the total processing times and weights of all early jobs removed in this way. 
Finally, in the vertex selection gadget, we delete from the constants~$L(i)$ and~$S(i)$ all terms~$N^{F(j, i)}$, $N^{F(i,j) +1}$, $N^{G(j, i)}$, or $N^{G(j, i)}$ with $\{i,j\} \notin E(H)$.
The proof of correctness is then analogous to \Cref{sec:PSharp,sec:WSharp}, resulting in \Cref{cor:eth-hardness}:
\ethhard*

\section{Conclusions}

In the current paper we completely resolved the parameterized complexity status of \prob with respect to parameters $p_{\#}$, $w_{\#}$, and $d_{\#}$. Our result also gives almost ETH tight bounds in the case when only one of $p_{\#}$ or $w_{\#}$ is bounded by a constant. However, there still remains several research directions to explore regarding the \prob problem, and its variants. Below we list a few questions that still remain open:

\begin{itemize}
\item Can the gap between lower and upper bound in \cref{cor:eth-hardness} be closed? That is, can one show a lower bound of $n^{o(k)}$ or can \prob be solved in $n^{O(k/ \lg k)}$ time, for $k=p_{\#}$ or $k=w_{\#}$?
\item The current FPT algorithms solving \prob for parameters $k=p_{\#}+w_{\#}$, $k=p_{\#}+d_{\#}$, or~$k=w_{\#}+d_{\#}$ have running times of the form~$2^{O(k\lg \lg k)} \cdot n^{O(1)}$ using the recent ILP-algorithm by Reis and Rothvoss~\cite{ReisR23}. Can any of these runtimes be improved to $2^{O(k)} \cdot n$, or can we a~$2^{\Omega (k\lg \lg k)} \cdot n^{O(1)}$  lower-bound?
\item Our result shows that \prob is $W[1]$-hard with respect to parameters~$p_{\#}$ and~$w_{\#}$, but it does not show that the problem is \emph{in} $W[1]$ for any of these parameters. Is \prob contained in $W[t]$ for some $t \geq 1$?
\end{itemize}

\bibliography{biblio}
  
\end{document}